\documentclass[11pt]{article}

\usepackage[margin=1in]{geometry}
\usepackage{amsmath,amssymb,amsthm}
\usepackage{graphicx}
\usepackage{booktabs}
\usepackage{hyperref}
\usepackage{algorithm}
\usepackage{algpseudocode}
\usepackage{bm}
\usepackage{xcolor}
\usepackage[numbers,sort&compress]{natbib}

\newtheorem{proposition}{Proposition}
\newtheorem{lemma}{Lemma}

\newtheorem{remark}{Remark}

\newcommand{\R}{\mathbb{R}}
\newcommand{\bd}{\mathbf{d}}
\newcommand{\bx}{\mathbf{x}}
\newcommand{\bxi}{\bm{\xi}}
\newcommand{\bA}{\mathbf{A}}
\newcommand{\bB}{\mathbf{B}}
\newcommand{\bC}{\mathbf{C}}
\newcommand{\bD}{\mathbf{D}}
\newcommand{\bF}{\mathbf{F}}
\newcommand{\bS}{\mathbf{S}}
\newcommand{\bJ}{\mathbf{J}}
\newcommand{\bPi}{\mathbf{\Pi}}
\newcommand{\bQ}{\mathbf{Q}}
\newcommand{\bI}{\mathbf{I}}
\newcommand{\bZ}{\mathbf{Z}}
\newcommand{\bW}{\mathbf{W}}
\newcommand{\bM}{\mathbf{M}}
\newcommand{\bX}{\mathbf{X}}
\newcommand{\by}{\mathbf{y}}
\newcommand{\bP}{\mathbf{P}}
\newcommand{\bU}{\mathbf{U}}
\newcommand{\bV}{\mathbf{V}}
\newcommand{\bSigma}{\mathbf{\Sigma}}
\newcommand{\bc}{\mathbf{c}}
\newcommand{\bg}{\mathbf{g}}
\newcommand{\bH}{\mathbf{H}}
\newcommand{\bp}{\mathbf{p}}
\newcommand{\bz}{\mathbf{z}}
\newcommand{\pinv}{{\dagger}}
\DeclareMathOperator{\col}{col}
\DeclareMathOperator{\ran}{ran}
\DeclareMathOperator{\rank}{rank}
\DeclareMathOperator{\diag}{diag}
\DeclareMathOperator*{\argmin}{arg\,min}

\title{A Generalized Ridge Regression\\ and Convolutional LASSO}

\author{Shintaro Yoshizawa\\
\small Nagoya Mathematical and Information Science Research
\thanks{This paper is based on material originally presented by the
author, Shintaro Yoshizawa, at the research workshop ``Mathematics
and Development of Statistical Modeling and Computational
Algorithms'' (\emph{T\=okeiteki Moderiningu to Keisan Arugorizumu no
S\=uri to Tenkai}), Graduate School of Informatics, Nagoya
University, February 18, 2017. Correspondence:
\texttt{shintaro.yoshizawa.net@gmail.com}.}}

\date{}

\begin{document}
\maketitle

\begin{abstract}
The Hodrick--Prescott filter (HPF), the classical device for
separating a slowly varying trend from a noisy or cyclical time
series, is a ridge-regression estimator whose penalty matrix is a
discrete second-difference operator $\bD$. Because $\bD$ is
rank-deficient, the trend can be re-parametrized in infinitely many
equivalent ways through a choice of generalized inverse of $\bD$; we
call two natural choices, built from the Moore--Penrose pseudoinverse
and from an alternative (non-minimum-norm) right inverse, the
\emph{A-} and \emph{B-representations}. We give a full, self-contained
account of this duality: we derive the general Penrose-type
parametrization of all admissible representations, prove that the
extracted trend is representation-invariant (Proposition~1), derive
in closed form the Bregman-type divergence that measures how two
representations disagree when driven by the \emph{same} coefficient
vector (Proposition~2), and prove that this divergence between the
HPF trend and the naive least-squares (polynomial) trend vanishes as
the regularization parameter $\lambda\to\infty$ (Proposition~3). We
then turn to the non-differentiable $\ell_1$ analogue of the HPF (the
$\ell_1$ trend filter), which trades quadratic for sparsity-inducing
regularization at the cost of losing smoothness of the objective; we
construct, by mollifying $|\cdot|$ with a compactly supported
biweight kernel, an explicit closed-form $C^2$ (in fact $C^3$)
\emph{Convolutional LASSO} surrogate that recovers Newton-type
quadratic convergence while retaining the qualitative kink-detecting
behaviour of the $\ell_1$ penalty. All theoretical results are
verified numerically, and the Convolutional LASSO trend filter is
benchmarked against a standard ADMM/IRLS $\ell_1$ solver, on real,
publicly available daily closing-price data for a semiconductor-sector
stock, NVIDIA Corporation (NVDA), covering February 2013--February
2018 (the Kaggle ``S\&P~500 stock data'' set). The sparse trend
filter isolates a small number of genuine growth-regime changes in
the five-year log-price series, cleanly separating them from the
smooth curve produced by the plain $L_2$ Hodrick--Prescott filter,
and localizes them close to the widely reported post-earnings jump
of November 11, 2016 ($+29.8\%$ in a single session) and the
subsequent acceleration of NVIDIA's data-center/gaming-driven growth.
\end{abstract}

\tableofcontents

\section{Introduction}

\subsection{Background and motivation}

The Hodrick--Prescott filter \citep{hodrick1997postwar} is one of the
most widely used trend-extraction devices in applied macroeconomics,
and, viewed purely as an estimator, it is nothing but a ridge
regression in which the design matrix is the identity and the penalty
matrix is a second-difference (discrete curvature) operator. Given an
observed series $\bd=(d_1,\dots,d_N)^\top$, the HPF trend is
\begin{equation}
  \bx_{\mathrm{opt}}(\lambda)
  = \argmin_{\bx\in\R^N} \; f_\lambda(\bx), \qquad
  f_\lambda(\bx) = \|\bd-\bx\|_2^2 + \lambda\|\bD\bx\|_2^2 ,
  \label{eq:hpf-primal}
\end{equation}
where $\bD\in\R^{(N-2)\times N}$ is the second-difference matrix (each
row is $(\dots,0,1,-2,1,0,\dots)$) and $\lambda>0$ is the smoothing
parameter. Because $f_\lambda$ is strictly convex ($\bI_N+\lambda
\bD^\top\bD\succ 0$), the minimizer in \eqref{eq:hpf-primal} is unique
and given by the closed-form ridge solution
$\bx_{\mathrm{opt}}(\lambda)=(\bI_N+\lambda\bD^\top\bD)^{-1}\bd$.

A less obvious fact --- which is the starting point of this paper ---
is that \eqref{eq:hpf-primal} admits a family of exactly equivalent
\emph{reparametrized} representations
\begin{equation}
  \tilde f_\lambda(\bxi)
  = \|\bd-\bC\bxi\|_2^2 + \lambda\|\bJ\bxi\|_2^2 , \qquad
  \bx = \bC\bxi ,
  \label{eq:hpf-reparam}
\end{equation}
where $\bJ=(\mathbf 0_{(N-2)\times2}\mid \bI_{N-2})$ simply selects the
last $N-2$ coordinates of $\bxi$, and $\bC\in\R^{N\times N}$ is any
invertible matrix satisfying $\bD\bC=\bJ$. Such $\bC$ exist in
abundance because $\bD$ has a $2$-dimensional kernel (the space of
discrete linear functions) and hence a whole affine family of
generalized (\,right-\,) inverses. Two natural choices considered
below --- one built from the Moore--Penrose pseudoinverse $\bD^\pinv$
(the ``B-representation''), one from an alternative, non-minimum-norm
right inverse of $\bD$ (the ``A-representation'') --- were introduced,
alongside the general Penrose-type parametrization of \emph{all} such
$\bC$, in the original presentation on which this paper is based
\citep{yoshizawa2017gotemba}. The present paper (i) supplies complete
proofs of the invariance and divergence properties of this
representation family, (ii) makes the associated Bregman-type
divergence formula precise (Section~\ref{sec:divergence}), and (iii)
develops a companion contribution --- the \emph{Convolutional LASSO}
--- that upgrades the non-smooth $\ell_1$ analogue of the HPF into an
explicitly differentiable, Newton-solvable convex program.

\subsection{Why does the representation matter?}

At first sight, \eqref{eq:hpf-reparam} looks like an unnecessary
complication: since $\bC$ is invertible, $\bxi=\bC^{-1}\bx$ and the
reparametrized objective is literally $\tilde
f_\lambda(\bC^{-1}\bx)=f_\lambda(\bx)$ termwise (because $\bD\bC=\bJ$
makes $\bJ\bC^{-1}=\bD$), so the two problems are trivially the same
optimization problem written in different coordinates. The interest
is not in the value of the objective but in what happens \emph{when
the two representations are queried inconsistently} -- e.g.\ when a
practitioner fits or regularizes the coefficient vector $\bxi$
directly (as one does, for instance, in a Bayesian or empirical-Bayes
treatment of the trend/cycle decomposition, or when only a subset of
$\bxi$ is estimated and the rest is held fixed across representations).
Proposition~\ref{prop:divergence} below shows that plugging the
\emph{same} $\bxi$ into two different representations $\bA$ and $\bB$
produces two \emph{different} trends $\bx_1=\bA\bxi$, $\bx_2=\bB\bxi$,
whose discrepancy is exactly quantified by a Bregman divergence of the
quadratic $f_\lambda$. Understanding this discrepancy is a
prerequisite for using generalized-inverse parametrizations of the
HPF safely, e.g.\ in generalized/vector versions of the filter, in
state-space (Kalman-filter) implementations that carry $\bxi$ rather
than $\bx$ as the latent state, or in the sparsity-generalized
(LASSO) filters discussed in Section~\ref{sec:convlasso}, where the
choice of coordinate representation directly affects which
coefficients are penalized.

\subsection{Related work}

The Hodrick--Prescott filter and its ridge-regression interpretation
are standard; see \citet{hodrick1997postwar} for the original
business-cycle application. The Moore--Penrose pseudoinverse and the
general parametrization of all solutions of a linear matrix equation
used throughout Section~\ref{sec:representations} are due to
\citet{penrose1955generalized}. The non-smooth $\ell_1$
(``$\ell_1$ trend filter'' / ``LASSO-Whittaker'') generalization of
the HPF, which replaces $\|\bD\bx\|_2^2$ by $\|\bD\bx\|_1$ to obtain
piecewise-linear trends with a sparse set of slope changes, is
studied thoroughly by \citet{kim2009l1}; \citet{yamada2016tsvd} (in
the source presentation, ``Yamada--Yoon (2016)'') consider the
practical choice of its regularization parameter. Fast dual/augmented
Lagrangian solvers for such sparsity-penalized problems, including
the Dual Augmented Lagrangian (DAL) method that inspired the
optimization strategy used for comparison in Section~\ref{sec:conv-experiment},
are given by \citet{tomioka2009dal}. Our Convolutional LASSO
construction is an instance of the general Moreau/mollifier smoothing
strategy used throughout convex optimization
\citep{nesterov2005smooth,boyd2011admm}; the reference ADMM solver we
use as an accuracy benchmark follows \citet{boyd2011admm}. On the
application side, trend/cycle decomposition of asset-price and
macroeconomic series is the original motivating use of the
Hodrick--Prescott filter \citep{hodrick1997postwar}, and its
piecewise-linear $\ell_1$ generalization is explicitly illustrated on
financial and macroeconomic series by \citet{kim2009l1} themselves;
detecting a small number of genuine growth-regime changes (structural
breaks) in a stock's price trajectory, as opposed to smoothing away
all local curvature, is precisely the situation in which the sparse
$\ell_1$/Convolutional-LASSO penalty is preferable to the plain
quadratic HPF penalty. All financial-data computations use
\texttt{pandas} \citep{mckinney2010pandas} on the public Kaggle
``S\&P~500 stock data'' set \citep{camnugent2018kaggle}.

\subsection{Contributions and outline}

All of the theoretical results below (Propositions~\ref{prop:invariance}--\ref{prop:mollified})
make precise and prove claims that were stated, sketched, or posed as
open questions in the source presentation \citep{yoshizawa2017gotemba}.
The numerical experiments of Section~\ref{sec:experiments} are new,
additional material developed by the present author specifically for
this paper, undertaken to (i) verify these theoretical results on
real data and (ii) empirically characterize the Convolutional LASSO
proposed in \citep{yoshizawa2017gotemba}, on a class of real data
(daily stock prices) not considered in the original presentation; no
numerical experiments of any kind appear in \citep{yoshizawa2017gotemba}
itself.

\begin{itemize}
\item \textbf{Section~\ref{sec:hpf}} recalls the HPF as a ridge
  regression and fixes notation.
\item \textbf{Section~\ref{sec:representations}} derives the general
  (Penrose-type) family of representations $\bC=(\bC_1,\bC_2)$ of the
  HPF and specializes it to the A- and B-representations.
\item \textbf{Section~\ref{sec:invariance}} proves representation
  invariance of the extracted trend (Proposition~\ref{prop:invariance}).
\item \textbf{Section~\ref{sec:divergence}} derives, from first
  principles, the Bregman divergence between two representations
  driven by a common coefficient vector (Proposition~\ref{prop:divergence}),
  and its vanishing limit as $\lambda\to\infty$
  (Proposition~\ref{prop:limit}).
\item \textbf{Section~\ref{sec:convlasso}} constructs the
  Convolutional LASSO: an explicit, closed-form, $C^2$ mollified
  absolute value, its relation to the derivative-optimal ``ternary
  polynomial'' mollifier of \citet{gravot2007ternary}, its use as a
  sparsity-inducing but differentiable regularizer, and a globalized
  (Levenberg--Marquardt-damped) Newton algorithm for it.
\item \textbf{Section~\ref{sec:experiments}} verifies
  Propositions~\ref{prop:invariance}--\ref{prop:limit} numerically and
  benchmarks the Convolutional LASSO against a standard ADMM/IRLS
  $\ell_1$ trend filter, both on real, publicly available daily
  closing-price data for the semiconductor stock NVIDIA Corporation
  (NVDA), 2013--2018.
\item \textbf{Section~\ref{sec:discussion}} discusses limitations and
  the two open problems inherited from the original presentation: the
  relationship between $\lambda$ and the geometric shape of the
  trend, and the broader theoretical/applied scope of the
  Convolutional LASSO.
\end{itemize}

\section{The Hodrick--Prescott filter as ridge regression}
\label{sec:hpf}

Let $\bd\in\R^N$ be the observed series and let
\begin{equation}
  \bD =
  \begin{pmatrix}
    1 & -2 & 1 & & \\
      & 1 & -2 & 1 & \\
      & & \ddots & \ddots & \ddots \\
      & & & 1 & -2 & 1
  \end{pmatrix}
  \in\R^{(N-2)\times N}
\end{equation}
be the second-difference operator, so that $(\bD\bx)_i =
x_i-2x_{i+1}+x_{i+2}$. The kernel of $\bD$ is exactly the
$2$-dimensional space of discrete linear (affine) sequences,
\begin{equation}
  \ker(\bD)=\col(\bPi), \qquad
  \bPi = \big(\mathbf 1_N \mid (1,2,\dots,N)^\top\big) \in\R^{N\times2},
  \label{eq:Pi}
\end{equation}
and $\rank(\bD)=N-2$. The set of admissible regularization strengths
is $\Lambda=\{\lambda\in\R : \bI_N+\lambda\bD^\top\bD\succ0\}\supseteq
[0,\infty)$. The HPF trend is the unique minimizer
\eqref{eq:hpf-primal}, and it is a linear (in $\bd$), shrinkage-type
estimator:
\begin{equation}
  \bx_{\mathrm{opt}}(\lambda) = (\bI_N+\lambda\bD^\top\bD)^{-1}\bd .
  \label{eq:hpf-solution}
\end{equation}
For $\lambda=0$ the trend is the data itself; as $\lambda\to\infty$,
$\bx_{\mathrm{opt}}(\lambda)$ is forced towards $\ker(\bD)$, i.e.\
towards a discrete linear function of $i$ (Proposition~\ref{prop:limit}
below makes this precise).

\section{Generalized-inverse representations of the HPF}
\label{sec:representations}

\subsection{Right inverses of $\bD$ and the A/B representations}

Because $\bD$ has full row rank $N-2$, it possesses a right inverse:
any matrix $\bF\in\R^{N\times(N-2)}$ with $\bD\bF=\bI_{N-2}$. The
\emph{minimum-norm} right inverse is the Moore--Penrose pseudoinverse
\begin{equation}
  \bF = \bD^\pinv = \bD^\top(\bD\bD^\top)^{-1} .
  \label{eq:F}
\end{equation}
Any other right inverse differs from $\bF$ by an element of $\ker(\bD)$
applied on the left of an arbitrary $2\times(N-2)$ matrix: if $\bS$
satisfies $\bD\bS=\bI_{N-2}$ then, because $\bD\bPi=\mathbf 0$,
\begin{equation}
  \bS = \bF + \bPi\bM, \qquad \bM\in\R^{2\times(N-2)} \text{ arbitrary},
  \label{eq:S}
\end{equation}
is also a right inverse ($\bD\bS = \bD\bF+\bD\bPi\bM = \bI_{N-2}$), and
every right inverse of $\bD$ is of this form. Equation~\eqref{eq:S} is
the precise version of the ``A-representation'' matrix $\bS$
displayed only graphically (as a triangular array with a specific
integer pattern) in the source presentation: any concrete choice of
$\bM\neq\mathbf 0$ gives one particular non-minimum-norm right inverse
$\bS\neq\bF$.

We define the two representations
\begin{equation}
  \bA = (\bPi \mid \bS) \in\R^{N\times N}
  \quad(\text{``A-representation''}), \qquad
  \bB = (\bPi \mid \bF) \in\R^{N\times N}
  \quad(\text{``B-representation''}),
  \label{eq:AB}
\end{equation}
each pairing the trend basis $\bPi$ of $\ker(\bD)$ with a right
inverse of $\bD$ on its complementary block. Both satisfy the key
identity
\begin{equation}
  \bD\bA = \bD\bB = \bJ = (\mathbf 0_{(N-2)\times 2}\mid \bI_{N-2}),
  \label{eq:DC=J}
\end{equation}
since $\bD\bPi=\mathbf0$ and $\bD\bF=\bD\bS=\bI_{N-2}$ by construction.
Both are generically invertible $N\times N$ matrices (their columns
span $\col(\bPi)\oplus\ran(\bF)=\R^N$ because $\bD$ has trivial kernel
restricted to $\ran(\bF)^\perp$... more directly: $\det\bA\neq0,\det
\bB\neq0$ for generic $\bM$, which we verify numerically in
Section~\ref{sec:experiments}).

\subsection{The general representation and Penrose's theorem}
\label{sec:general-rep}

More generally, following the source presentation, consider any
$\bC=(\bC_1,\bC_2)$, $\bC_1\in\R^{N\times2}$, $\bC_2\in\R^{N\times(N-2)}$,
satisfying the reparametrization condition
$\bD\bC_1=\mathbf 0$, $\bD\bC_2=\bI_{N-2}$ (equivalently
$\bD\bC=\bJ$, i.e.\ \eqref{eq:DC=J}). Since $\bD\bC_1=\mathbf0$ forces
$\bC_1\in\col(\bPi)$, and since $\bD\bC_2=\bI_{N-2}$ is precisely
the generalized linear equation $\bD\bX=\bI_{N-2}$ solved for $\bX$,
Penrose's theorem on the general solution of $\bA\bX\bB=\bC$
\citep{penrose1955generalized} gives, taking $\bA=\bD$, $\bB=\bI$,
\begin{equation}
  \bC_1 = (\bI_N-\bD^\pinv\bD)\bZ, \quad \bZ\in\R^{N\times2}, \qquad
  \bC_2 = \bD^\pinv + (\bI_N-\bD^\pinv\bD)\bW, \quad
  \bW\in\R^{N\times(N-2)},
  \label{eq:general-C}
\end{equation}
where $\bI_N-\bD^\pinv\bD$ is the orthogonal projector onto $\ker(\bD)
=\col(\bPi)$. (This matches \eqref{eq:S}: $\bC_2=\bF+\Pi(\Pi^\top\Pi)^{-1}\Pi^\top\bW=\bF+\bPi\bM$
with $\bM=(\bPi^\top\bPi)^{-1}\bPi^\top\bW$.) The A- and
B-representations correspond, in the notation of
\eqref{eq:general-C}, to
\begin{center}
\begin{tabular}{lll}
\toprule
 & $\bZ$ & $\bW$ \\
\midrule
A-representation & $\bA\begin{pmatrix}\bI_2\\ \mathbf 0\end{pmatrix}$ & $\bA\bJ^\top$ \\[4pt]
B-representation & $\bB\begin{pmatrix}\bI_2\\ \mathbf 0\end{pmatrix}$ & $\mathbf 0,\ \bD^\pinv,\ \text{or }\bD^\top$ \\
\bottomrule
\end{tabular}
\end{center}
i.e.\ the B-representation is recovered by the trivial choice
$\bW=\mathbf0$ (giving back $\bF=\bD^\pinv$ exactly), and the
A-representation by a non-trivial $\bW$, consistently with
\eqref{eq:S}--\eqref{eq:general-C}. (Other admissible choices, such as
$\bW=\bD^\pinv$ or $\bW=\bD^\top$, mentioned in the source
presentation, give still further, distinct representations; we use a
single generic $\bW\ne 0$, drawn once at random, for the numerical
A-representation in Section~\ref{sec:experiments}, since the
qualitative behaviour of Propositions~\ref{prop:invariance}--\ref{prop:limit}
does not depend on which non-trivial $\bW$ is chosen.)

\section{Representation invariance of the HPF trend}
\label{sec:invariance}

\begin{proposition}[Representation invariance]
\label{prop:invariance}
Let $\bC\in\R^{N\times N}$ be any invertible matrix satisfying
$\bD\bC=\bJ$ as in \eqref{eq:DC=J}, and let
\begin{equation}
  \bxi_{\mathrm{opt}}^{\bC}(\lambda)
  = \argmin_{\bxi\in\R^N} \tilde f_\lambda^{\bC}(\bxi), \qquad
  \tilde f_\lambda^{\bC}(\bxi) = \|\bd-\bC\bxi\|_2^2+\lambda\|\bJ\bxi\|_2^2 .
\end{equation}
Then $\bC\bxi^{\bC}_{\mathrm{opt}}(\lambda)=\bx_{\mathrm{opt}}(\lambda)$
for every admissible $\bC$; in particular the A- and B-representations
give the same trend, $\bA\bxi^{\bA}_{\mathrm{opt}}(\lambda) =
\bB\bxi^{\bB}_{\mathrm{opt}}(\lambda) = \bx_{\mathrm{opt}}(\lambda)$.
\end{proposition}

\begin{proof}
Since $\bC$ is invertible, $\bxi=\bC^{-1}\bx$ is a bijection of $\R^N$
onto itself, and $\bD\bC=\bJ$ implies $\bJ\bC^{-1}=\bD$ (right-multiply
both sides of $\bD\bC=\bJ$ by $\bC^{-1}$). Hence for every $\bx\in\R^N$,
writing $\bxi=\bC^{-1}\bx$,
\[
  \tilde f_\lambda^{\bC}(\bC^{-1}\bx)
  = \|\bd-\bC\bC^{-1}\bx\|_2^2 + \lambda\|\bJ\bC^{-1}\bx\|_2^2
  = \|\bd-\bx\|_2^2+\lambda\|\bD\bx\|_2^2
  = f_\lambda(\bx) .
\]
Thus $\tilde f_\lambda^{\bC}\circ\bC^{-1}\equiv f_\lambda$ on $\R^N$,
so the two strictly convex problems have corresponding unique
minimizers under the bijection $\bx=\bC\bxi$: $\bxi^{\bC}_{\mathrm{opt}}
(\lambda)=\bC^{-1}\bx_{\mathrm{opt}}(\lambda)$, i.e.\
$\bC\bxi^{\bC}_{\mathrm{opt}}(\lambda)=\bx_{\mathrm{opt}}(\lambda)$, as
claimed.
\end{proof}

\begin{remark}
Proposition~\ref{prop:invariance} is the precise statement behind the
commutative diagram (``Proposition 1,'' \emph{meidai 1}) of the source presentation: the map
$\phi_{\bC}:\bxi\mapsto\bC\bxi$ intertwines the reparametrized problem
$(\R^N,\tilde f_\lambda^{\bC})$ with the original problem
$(\R^N,f_\lambda)$, and this holds \emph{for every} admissible $\bC$
--- the choice of generalized inverse used to build $\bC$ is
therefore immaterial as long as one is consistent about which $\bC$
generated the coefficient vector $\bxi$ being manipulated. What Is
\emph{not} invariant is the vector $\bxi$ itself: $\bxi^{\bA}_{\mathrm{opt}}
(\lambda)\ne\bxi^{\bB}_{\mathrm{opt}}(\lambda)$ in general (only their
images $\bA\bxi^{\bA}_{\mathrm{opt}}=\bB\bxi^{\bB}_{\mathrm{opt}}$
coincide). This is verified numerically to machine precision in
Section~\ref{sec:num-duality}.
\end{remark}

\section{The Bregman divergence between two representations}
\label{sec:divergence}

Proposition~\ref{prop:invariance} shows that solving \emph{separately}
in each representation and mapping back gives the same trend. A
different, and more delicate, question -- the one raised in the
source presentation as ``the divergence of the A/B representations'' (the
divergence between the A- and B-representations) -- is: if the
\emph{same} coefficient vector $\bxi$ is mapped through the two
representations, $\bx_1=\bA\bxi$ and $\bx_2=\bB\bxi$, how different
are the resulting trends?

\subsection{Bregman divergence of a quadratic}

Since $f_\lambda$ is a smooth, strictly convex quadratic,
\begin{equation}
  f_\lambda(\bx) = \bx^\top\bQ_\lambda\bx - 2\bd^\top\bx+\bd^\top\bd,
  \qquad \bQ_\lambda := \bI_N+\lambda\bD^\top\bD \succ 0,
  \label{eq:f-quadratic}
\end{equation}
its Legendre--Fenchel conjugate is itself a quadratic,
$f_\lambda^{*}(\by)=\tfrac14(\by+2\bd)^\top\bQ_\lambda^{-1}(\by+2\bd)-\bd^\top\bd$,
and the associated Bregman divergence
$D_\lambda(\bx_1\|\bx_2):=f_\lambda(\bx_1)-f_\lambda(\bx_2)-\nabla
f_\lambda(\bx_2)^\top(\bx_1-\bx_2)$ (equivalently, in the
primal--dual form used in the source presentation, $D_\lambda(\bx_1\|
\bx_2)=f_\lambda(\bx_1)+f_\lambda^{*}(\by(\bx_2))-\bx_1^\top\by(\bx_2)$
with $\by(\bx_2)=\nabla f_\lambda(\bx_2)$) reduces, for a quadratic, to
the associated squared Mahalanobis distance:

\begin{lemma}
\label{lem:bregman-quadratic}
For $f_\lambda$ as in \eqref{eq:f-quadratic},
\begin{equation}
  D_\lambda(\bx_1\|\bx_2) = (\bx_1-\bx_2)^\top \bQ_\lambda (\bx_1-\bx_2)
  = \|\bx_1-\bx_2\|_{\bQ_\lambda}^2 .
  \label{eq:bregman-formula}
\end{equation}
\end{lemma}
\begin{proof}
$\nabla f_\lambda(\bx)=2\bQ_\lambda\bx-2\bd$. Substituting into the
definition,
\begin{align*}
D_\lambda(\bx_1\|\bx_2)
&= \big(\bx_1^\top\bQ_\lambda\bx_1-2\bd^\top\bx_1\big)
 - \big(\bx_2^\top\bQ_\lambda\bx_2-2\bd^\top\bx_2\big)
 - (2\bQ_\lambda\bx_2-2\bd)^\top(\bx_1-\bx_2)\\
&= \bx_1^\top\bQ_\lambda\bx_1 - 2\bx_2^\top\bQ_\lambda\bx_1
   + \bx_2^\top\bQ_\lambda\bx_2
\;=\; (\bx_1-\bx_2)^\top\bQ_\lambda(\bx_1-\bx_2),
\end{align*}
using symmetry of $\bQ_\lambda$ and cancellation of all terms linear
in $\bd$.
\end{proof}

\subsection{Representation divergence}

\begin{proposition}[Representation divergence]
\label{prop:divergence}
Let $\bxi\in\R^N$ be any fixed coefficient vector and set
$\bx_1=\bA\bxi$, $\bx_2=\bB\bxi$. Then
\begin{equation}
  D_\lambda(\bx_1\|\bx_2)
  = \big\|(\bA-\bB)\bxi\big\|_{\bQ_\lambda}^2
  = \big[(\bA-\bB)\bxi\big]^\top(\bI_N+\lambda\bD^\top\bD)\big[(\bA-\bB)\bxi\big] .
  \label{eq:prop2}
\end{equation}
In particular $D_\lambda(\bx_1\|\bx_2)=0$ if and only if
$(\bA-\bB)\bxi\in\ker(\bQ_\lambda)=\{\mathbf0\}$, i.e.\ if and only if
$\bA\bxi=\bB\bxi$.
\end{proposition}
\begin{proof}
Immediate from Lemma~\ref{lem:bregman-quadratic} with
$\bx_1-\bx_2=(\bA-\bB)\bxi$.
\end{proof}

\begin{remark}
The source presentation states the simpler-looking identity
$D_\lambda(\bx_1\|\bx_2)=\|(\bA-\bB)\bxi\|_2^2$ (plain Euclidean norm).
Equation~\eqref{eq:prop2} is the mathematically precise version: the
correct norm weighting the discrepancy $(\bA-\bB)\bxi$ is the
regularization-dependent metric $\bQ_\lambda=\bI_N+\lambda\bD^\top\bD$
induced by $f_\lambda$ itself, not the ambient Euclidean metric; the
two coincide only in the unregularized case $\lambda=0$. This
distinction matters quantitatively (the two differ by orders of
magnitude for large $\lambda$, as $\bD^\top\bD$ can have eigenvalues
much larger than $1$) and is confirmed numerically in
Section~\ref{sec:num-duality} (Table~\ref{tab:prop2}), where both
sides of \eqref{eq:prop2} are computed independently and agree to
machine precision, while the plain Euclidean norm does not.
\end{remark}

\subsection{Vanishing divergence in the $\lambda\to\infty$ limit}

The source presentation's Proposition~3 addresses a different,
asymptotic, divergence: between the full HPF trend
$\bx_1(\lambda)=\bx_{\mathrm{opt}}(\lambda)$ and the ``pure trend
component'' $\bx_2$ obtained by ordinary least-squares projection of
$\bd$ onto the trend basis $\col(\bPi)$ alone, i.e.\ the (unregularized)
answer one would get by fitting only a discrete linear function to the
data:
\begin{equation}
  \bx_2 = \bP_{\bPi}\bd, \qquad
  \bP_{\bPi} = \bPi(\bPi^\top\bPi)^{-1}\bPi^\top
  \quad(\text{the orthogonal projector onto }\col(\bPi)).
  \label{eq:x2}
\end{equation}

\begin{proposition}[Vanishing divergence as $\lambda\to\infty$]
\label{prop:limit}
With $\bx_1(\lambda)=(\bI_N+\lambda\bD^\top\bD)^{-1}\bd$ and $\bx_2$ as
in \eqref{eq:x2},
\begin{equation}
  \lim_{\lambda\to\infty} D_\lambda\big(\bx_1(\lambda)\,\big\|\,\bx_2\big) = 0 .
  \label{eq:prop3}
\end{equation}
Moreover $\|\bx_1(\lambda)-\bx_2\|_2\to0$ as $\lambda\to\infty$, i.e.\
the HPF trend converges (not merely in the $\bQ_\lambda$-divergence,
but in ordinary Euclidean distance) to the OLS linear trend as the
smoothing parameter grows without bound.
\end{proposition}
\begin{proof}
Let $\bD=\bU\bSigma\bV^\top$ be a (thin) singular value decomposition,
$\bU\in\R^{(N-2)\times(N-2)}$, $\bV\in\R^{N\times(N-2)}$ orthonormal,
$\bSigma=\diag(\sigma_1,\dots,\sigma_{N-2})\succ0$ (all singular
values of $\bD$ are strictly positive since $\rank\bD=N-2$). Extend
$\bV$ to an orthonormal basis $(\bV_0\mid \bV)$ of $\R^N$ where
$\bV_0\in\R^{N\times2}$ is an orthonormal basis of $\ker(\bD)=\col
(\bPi)$ (so $\bP_{\bPi}=\bV_0\bV_0^\top$). Then
$\bD^\top\bD=\bV\bSigma^2\bV^\top$, and
\[
  \bQ_\lambda = \bI_N+\lambda\bD^\top\bD
  = \bV_0\bV_0^\top + \bV(\bI_{N-2}+\lambda\bSigma^2)\bV^\top ,
\]
so
\[
  \bQ_\lambda^{-1} = \bV_0\bV_0^\top + \bV(\bI_{N-2}+\lambda\bSigma^2)^{-1}\bV^\top ,
\]
and therefore
\[
  \bx_1(\lambda) = \bQ_\lambda^{-1}\bd
  = \underbrace{\bV_0\bV_0^\top\bd}_{=\bx_2}
  + \bV(\bI_{N-2}+\lambda\bSigma^2)^{-1}\bV^\top\bd .
\]
Hence
\[
  \bx_1(\lambda)-\bx_2 = \bV(\bI_{N-2}+\lambda\bSigma^2)^{-1}\bV^\top\bd,
  \qquad
  \|\bx_1(\lambda)-\bx_2\|_2
  \le \big\|(\bI_{N-2}+\lambda\bSigma^2)^{-1}\big\|_{\mathrm{op}}\,\|\bd\|_2
  = \frac{\|\bd\|_2}{1+\lambda\sigma_{\min}^2},
\]
so $\|\bx_1(\lambda)-\bx_2\|_2\xrightarrow[\lambda\to\infty]{}0$,
proving the Euclidean convergence claim. For the divergence,
Lemma~\ref{lem:bregman-quadratic} gives
$D_\lambda(\bx_1(\lambda)\|\bx_2)=\|\bx_1(\lambda)-\bx_2\|_{\bQ_\lambda}^2$;
writing $\bc:=\bV^\top\bd$ and using the eigendecomposition above,
\[
  D_\lambda(\bx_1(\lambda)\|\bx_2)
  = \bc^\top(\bI_{N-2}+\lambda\bSigma^2)^{-1}(\bI_{N-2}+\lambda\bSigma^2)
    (\bI_{N-2}+\lambda\bSigma^2)^{-1}\bc
  = \sum_{k=1}^{N-2} \frac{c_k^2}{1+\lambda\sigma_k^2}
  \xrightarrow[\lambda\to\infty]{} 0
\]
by dominated convergence (each summand is bounded by $c_k^2$ and
tends to $0$), which proves \eqref{eq:prop3}.
\end{proof}

\begin{remark}
The rate is explicit and geometric: $D_\lambda(\bx_1(\lambda)\|\bx_2)
=O(1/\lambda)$, governed by the smallest singular value $\sigma_{\min}
(\bD)$ of the second-difference operator. This is exactly the
behaviour observed numerically in Figure~\ref{fig:prop3}
(Section~\ref{sec:num-duality}): the divergence is essentially flat
for small-to-moderate $\lambda$ and then decays as a power law once
$\lambda\sigma_{\min}^2\gtrsim1$.
\end{remark}

\section{From the $\ell_1$ trend filter to the Convolutional LASSO}
\label{sec:convlasso}

\subsection{The non-smoothness problem}

Replacing the quadratic penalty $\lambda\|\bD\bx\|_2^2$ in
\eqref{eq:hpf-primal} by $\lambda\|\bD\bx\|_1=\lambda\sum_i|(\bD\bx)_i|$
gives the ($\ell_1$) trend filter
\citep{kim2009l1,yamada2016tsvd},
\begin{equation}
  \min_{\bx\in\R^N} \; \|\bd-\bx\|_2^2 + \lambda\|\bD\bx\|_1 .
  \label{eq:l1-tf}
\end{equation}
Because the second difference is now penalized in $\ell_1$ rather
than $\ell_2$, the minimizer is typically \emph{piecewise linear} with
a small number of slope changes (``kinks''), located exactly where
$(\bD\bx)_i\ne0$: the $\ell_1$ penalty induces sparsity in the
\emph{curvature} of the trend, in the same way the ordinary LASSO
induces sparsity in a coefficient vector. This is attractive for
applications such as changepoint/regime detection, but the objective
in \eqref{eq:l1-tf} is not differentiable wherever $(\bD\bx)_i=0$
--- precisely at the points that matter most (the flat, purely-linear
segments) --- so Newton's method cannot be applied directly, and one
must resort to interior-point, ADMM, IRLS, or dual/DAL-type solvers
\citep{kim2009l1,tomioka2009dal,boyd2011admm}.

\subsection{Mollification: an explicit, closed-form smoothing}
\label{sec:mollification}

We regularize $\rho(t)=|t|$ by convolving it with a compactly
supported mollifier $\psi_\varepsilon$, i.e.\ we replace $\rho$ by
\begin{equation}
  \tilde\rho_\varepsilon(x) := (\psi_\varepsilon * \rho)(x)
  = \int_{-\infty}^{\infty}\psi_\varepsilon(x-y)\,\rho(y)\,dy ,
  \label{eq:mollify}
\end{equation}
which is exactly the construction sketched (with a generic bump
function) in the source presentation. Any smooth, symmetric,
compactly supported, unit-mass $\psi_\varepsilon$ works; we choose the
\emph{biweight (quartic) kernel}
\begin{equation}
  \psi_\varepsilon(t) =
  \begin{cases}
    \dfrac{15}{16\varepsilon}\Big(1-\big(t/\varepsilon\big)^2\Big)^2, & |t|\le\varepsilon, \\[6pt]
    0, & |t|>\varepsilon,
  \end{cases}
  \label{eq:biweight}
\end{equation}
because it is itself $C^1$ (vanishing together with its first
derivative at $t=\pm\varepsilon$), which upgrades the usual
``mollification gains one derivative'' folklore ($\rho*\psi$ inherits
one more derivative than $\min(\rho,\psi)$) into an especially clean,
\emph{globally $C^3$}, result. Carrying out the convolution
\eqref{eq:mollify} in closed form (elementary but tedious; done here
with a computer-algebra system for exactness) gives, using the
evenness of $\tilde\rho_\varepsilon$ and matching $\rho(x)=x$ for $x\ge\varepsilon$:

\begin{proposition}[Closed-form Convolutional LASSO kernel]
\label{prop:mollified}
For the biweight kernel \eqref{eq:biweight},
\begin{equation}
  \tilde\rho_\varepsilon(x) =
  \begin{cases}
    \dfrac{5\varepsilon}{16} + \dfrac{15x^2}{16\varepsilon}
    - \dfrac{5x^4}{16\varepsilon^3} + \dfrac{x^6}{16\varepsilon^5}, & |x|<\varepsilon,\\[8pt]
    |x|, & |x|\ge\varepsilon .
  \end{cases}
  \label{eq:rho-tilde}
\end{equation}
Moreover $\tilde\rho_\varepsilon\in C^3(\R)$, is convex on all of $\R$,
and matches $\rho(x)=|x|$, together with its first three derivatives,
exactly at the junction points $x=\pm\varepsilon$:
\begin{equation}
  \tilde\rho_\varepsilon(0)=\tfrac{5}{16}\varepsilon>0, \quad
  \tilde\rho_\varepsilon'(0)=0, \quad
  \tilde\rho_\varepsilon''(0)=\tfrac{15}{8\varepsilon}>0, \qquad
  \tilde\rho_\varepsilon''(x) = 2\psi_\varepsilon(x)\ge0 .
  \label{eq:rho-tilde-derivs}
\end{equation}
\end{proposition}

\begin{proof}[Proof (sketch)]
By symmetry it suffices to compute \eqref{eq:mollify} for $0\le
x<\varepsilon$, splitting the integral at $t=x$ (since $\psi_\varepsilon$ has support
$[-\varepsilon,\varepsilon]$ and $|x-t|$ changes sign there):
\[
  \tilde\rho_\varepsilon(x) = \int_{-\varepsilon}^{x}\psi_\varepsilon(t)(x-t)\,dt
  + \int_{x}^{\varepsilon}\psi_\varepsilon(t)(t-x)\,dt .
\]
Expanding $\psi_\varepsilon$ from \eqref{eq:biweight} as a quartic
polynomial in $t$ and integrating termwise (elementary polynomial
integration; verified symbolically) yields exactly
\eqref{eq:rho-tilde}. For $x\ge\varepsilon$, $|x-t|=x-t$ for every
$t\in[-\varepsilon,\varepsilon]$ (since $t\le\varepsilon\le x$), so
$\tilde\rho_\varepsilon(x)=x\int\psi_\varepsilon-\int t\,\psi_\varepsilon(t)\,dt
= x\cdot1-0=x$, using unit mass and evenness of $\psi_\varepsilon$.
The derivatives in \eqref{eq:rho-tilde-derivs} follow by direct
differentiation of \eqref{eq:rho-tilde} and evaluating at $0$ and
$\varepsilon$; matching at $x=\varepsilon$ to third order is a
consequence of $\psi_\varepsilon\in C^1$ with $\psi_\varepsilon(\pm
\varepsilon)=\psi_\varepsilon'(\pm\varepsilon)=0$. Convexity follows
since $\tilde\rho_\varepsilon''(x)=2\psi_\varepsilon(x)\ge0$
everywhere ($\rho''=2\delta_0$ as a distribution, and mollifying a
convex function with a nonnegative kernel preserves convexity).
\end{proof}

Figure~\ref{fig:convergence} (left panel) shows $\tilde\rho_\varepsilon$ for several
values of $\varepsilon$: as $\varepsilon\to0$, $\tilde\rho_\varepsilon
\to\rho$ pointwise (indeed uniformly, at rate $O(\varepsilon)$, since
$\tilde\rho_\varepsilon(0)=\tfrac5{16}\varepsilon$ and both functions
agree exactly outside $(-\varepsilon,\varepsilon)$), while for every
$\varepsilon>0$, $\tilde\rho_\varepsilon$ is a genuine, globally
convex, $C^3$ function with strictly positive curvature at the
origin.

\subsection{An alternative, derivative-optimal mollifier: the ``ternary polynomial''}
\label{sec:ternary}

The biweight kernel \eqref{eq:biweight} was chosen above purely for
its algebraic simplicity: it is the lowest-degree even polynomial
bump that is itself $C^1$, which was enough to make $\tilde\rho_\varepsilon$
globally $C^3$ in closed form (Proposition~\ref{prop:mollified}).
It is worth pointing out explicitly that this is one convenient
choice among an entire family of compactly supported mollifiers that
could be substituted into \eqref{eq:mollify}, and that a
considerably more elaborate, \emph{optimal} member of this family
--- the \emph{ternary polynomial} of \citet{gravot2007ternary} ---
was developed by F.\ Gravot, Y.\ Hirano, and S.\ Yoshizawa for an
unrelated application, namely time-optimal smoothing of robot
motion-planning trajectories, and is structurally the same
convolution-mollification idea used here.

\citet{gravot2007ternary} construct, for any prescribed bounds
$M_0,M_1,\dots,M_n$ on a function and its first $n$ derivatives, a
compactly supported function $T_n$ on an interval $[-a_n/2,a_n/2]$
satisfying $T_n(0)=M_0$, $T_n^{(k)}(\pm a_n/2)=0$ for $k=0,\dots,n-1$,
and $|T_n^{(k)}|\le M_k$ for $k=0,\dots,n$, with the support half-width
$a_n$ made \emph{as small as possible}. $T_n$ is built recursively
from the indicator (rectangular) function $T_0$ of an interval by
alternating a symmetric finite-difference step with an integration
step -- schematically, $T_k$ is obtained by shifting and differencing
$T_{k-1}$ and integrating the result -- so that $T_k^{(n-k)}$ is
piecewise constant, $T_k^{(n-k-1)}$ is piecewise linear, and so on up
to $T_n$ itself, which is $C^{n-1}$ and piecewise polynomial of
degree $n$. The interval lengths $a_0<a_1<\dots<a_n$ that determine
$T_n$ are themselves obtained by successively maximizing each $a_i$
subject to the derivative bounds, each maximization reducing to
finding the largest positive root of a single explicit polynomial
(degree $n(n-1)$ in the worst case, but only quartic for $n=3$, the
jerk-bounded case relevant to robot control); the construction is
therefore fully explicit and non-iterative once these roots are
found, and $T_n$ is, in this precise derivative-bounded sense, the
\emph{fastest} (minimal-support) bump function attaining the
prescribed smoothness and derivative bounds.

The connection to Section~\ref{sec:convlasso}'s Convolutional LASSO
is direct, not merely thematic. \citet{gravot2007ternary} use $T_n$
itself (rescaled) as the mollifying kernel $k$ in a convolution
$\mathrm{path}(p)=(p*k)(2p-1)$ that smooths a two-segment,
piecewise-linear robot path with a corner at a waypoint $q_i$ into a
$C^{n-1}$ trajectory joining the neighbouring waypoints $q_{i-1}$ and
$q_{i+1}$; they show (their Eq.~45, in the present notation)
\begin{equation}
  (\mathrm{path}*k)''(s) = (q_{i+1}-2q_i+q_{i-1})\cdot k(s),
  \label{eq:gravot-corner}
\end{equation}
i.e.\ the second derivative of the smoothed path is the corner's
\emph{discrete second difference} $q_{i+1}-2q_i+q_{i-1}$ -- exactly
one row of the operator $\bD$ used throughout this paper -- multiplied
by the kernel itself. Equation~\eqref{eq:gravot-corner} is the
trajectory-smoothing counterpart of the identity
$\tilde\rho_\varepsilon''=2\psi_\varepsilon$ used in the proof of
Proposition~\ref{prop:mollified}: in both cases, mollifying a
function whose second derivative is a discrete impulse (a corner, or
$\rho''=2\delta_0$) reproduces the mollifying kernel itself, scaled
by the impulse's magnitude. Where \citet{gravot2007ternary} mollify a
\emph{single} corner of a robot path with an interval-supported,
derivative-optimal $T_n$, Section~\ref{sec:convlasso} mollifies the
absolute-value penalty applied to \emph{every} entry of $\bD\bx$
simultaneously, but the underlying convolution identity is the same
one.

The two constructions also make different trade-offs, worth stating
explicitly since Section~\ref{sec:mollification} chose the simpler of
the two. The ternary polynomial $T_n$ is optimal in the sense of
minimal support for given derivative bounds, and was designed
precisely because an earlier, non-optimal choice used in the same
line of work -- the classical $C^\infty$ bump
$k(x)=\exp(-1/(1-x^2))/c$, attributed by \citet{gravot2007ternary} to
their own earlier work \citep{hirano2005rrt} -- has no
closed form (the smoothed trajectory must be obtained by numerical
integration) and offers no direct control over the derivative bounds
of the result. The biweight kernel \eqref{eq:biweight} used here
inherits the second drawback only partially (Proposition~\ref{prop:mollified}
gives explicit, if kernel-specific, bounds $\tilde\rho_\varepsilon''\le
15/(8\varepsilon)$) but not the first: it was chosen specifically so
that the convolution \eqref{eq:mollify} could be carried out in
closed form directly, without the recursive root-finding construction
that produces $T_n$. Substituting a rescaled ternary polynomial
$T_{n-2}$ for the biweight kernel $\psi_\varepsilon$ in
\eqref{eq:mollify} would yield a Convolutional LASSO surrogate
$\tilde\rho_\varepsilon\in C^{n-1}$ of arbitrary prescribed smoothness
order with the minimal-support (fastest-transition) property in
place of \eqref{eq:rho-tilde}'s fixed quartic-degree transition; we
leave a full development of this higher-order, support-optimal
Convolutional LASSO to future work
(Section~\ref{sec:discussion}).

\paragraph{A concrete worked example: $T_0\to T_1\to T_2\to T_3$.}
To make the recursive construction of $T_n$ tangible, we work through
it explicitly for the simplest choice of interval increments, namely
\emph{equal} spacing $a_k-a_{k-1}\equiv1$ for every $k$ (as opposed to
the generally \emph{unequal}, derivative-bound-optimal spacing that
\citet{gravot2007ternary} compute via the root-finding procedure of
their Eqs.~21--26). With equal spacing, iterating the
``symmetric-difference-then-integrate'' rule (their Eq.~16) reduces
to repeated convolution with a fixed unit-width box function
$\mathrm{box}(x)$, equal to $1$ for $|x|<1/2$ and $0$ otherwise, and $T_k$ becomes exactly
the classical order-$(k+1)$ cardinal B-spline:
\begin{equation}
  T_0 = \mathrm{box}, \qquad
  T_k = T_{k-1} * \mathrm{box} \quad (k=1,2,3,\dots),
  \label{eq:ternary-worked}
\end{equation}
with the following explicit closed forms (verified here by direct
numerical convolution against \eqref{eq:ternary-worked}, matching to
within discretization error):
\begin{align}
  T_0(x) &= \begin{cases}1,&|x|<\tfrac12,\\0,&|x|\ge\tfrac12,\end{cases}
  \qquad\qquad
  T_1(x) = \begin{cases}1-|x|,&|x|<1,\\0,&|x|\ge1,\end{cases}\\[4pt]
  T_2(x) &= \begin{cases}
    \tfrac34-x^2, & |x|\le\tfrac12,\\
    \tfrac12\big(\tfrac32-|x|\big)^2, & \tfrac12<|x|<\tfrac32,\\
    0, & |x|\ge\tfrac32,
  \end{cases}
  \qquad
  T_3(x) = \begin{cases}
    \tfrac23-x^2+\tfrac{|x|^3}2, & |x|\le1,\\
    \tfrac{(2-|x|)^3}6, & 1<|x|<2,\\
    0, & |x|\ge2.
  \end{cases}
  \label{eq:ternary-closedform}
\end{align}
Each successive convolution with the box both widens the support by
$1$ (from $a_0=1$ up to $a_1=2,a_2=3,a_3=4$) and raises the
smoothness by one order: $T_0$ is merely bounded ($C^{-1}$, a jump
discontinuity), $T_1$ is continuous but has a corner at $x=0$ ($C^0$,
a triangle/tent function), $T_2$ is $C^1$ (a smooth-looking but only
once-differentiable bump), and $T_3$ is $C^2$ -- the cubic B-spline,
with $T_3(0)=\tfrac23$, matching $\tilde\rho_\varepsilon$'s own
target smoothness order in Proposition~\ref{prop:mollified}.
Figure~\ref{fig:ternary} (left) plots all four functions on a common
axis, visibly illustrating the progressive rounding of the corner at
the origin as $k$ increases; the right panel plots $T_3$ together
with its first two derivatives, confirming numerically that $T_3'$
and $T_3''$ are both continuous (in particular $T_3''$, though
continuous, has visible corners at the internal knots $x=\pm1$ and at
the support boundary $x=\pm2$, since the third derivative -- not
shown -- is genuinely discontinuous there, exactly as expected for a
$C^2$, not $C^3$, function). This equal-spacing example is only the
simplest member of the family: the actual optimal $T_n$ of
\citet{gravot2007ternary}, with the $a_i$ chosen by their root-finding
recursion rather than fixed to increments of $1$, achieves the same
$C^{n-1}$ smoothness with the smallest possible support width $a_n$
for independently \emph{prescribed} bounds $M_0,\dots,M_n$ on each
derivative -- a property the equal-spacing construction above does
not have, since here every derivative bound is whatever the fixed
box-convolution happens to produce (e.g.\ $\max|T_3'|=2/3$,
$\max|T_3''|=2$, read off Figure~\ref{fig:ternary}, rather than
independently chosen targets).

\begin{figure}[htbp]
  \centering
  \includegraphics[width=0.48\linewidth]{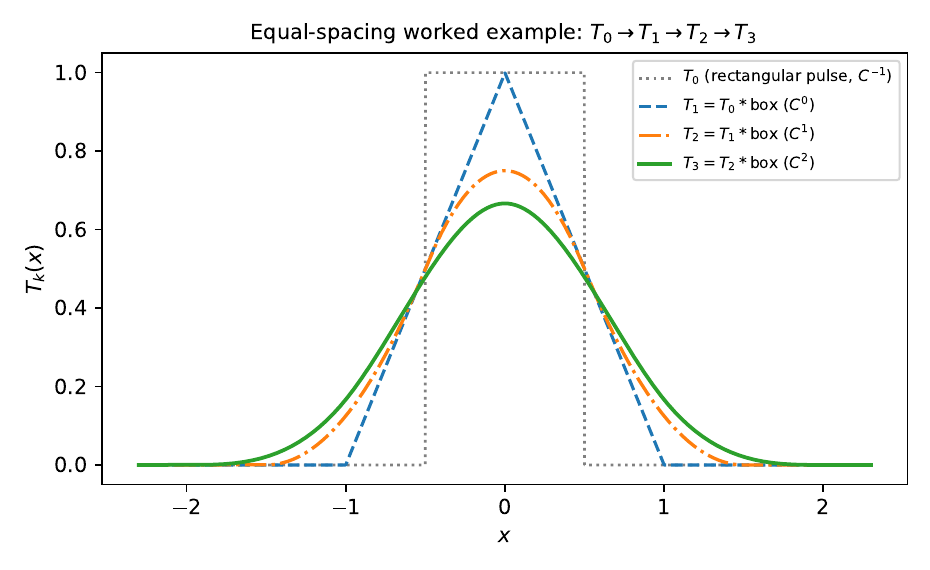}\quad
  \includegraphics[width=0.46\linewidth]{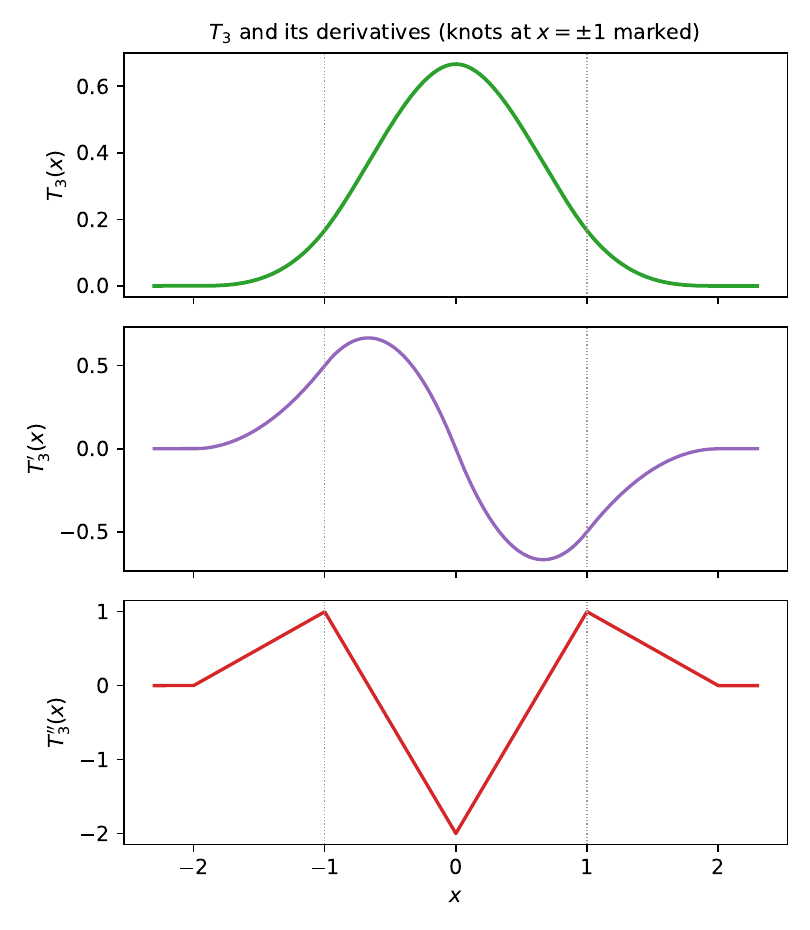}
  \caption{A concrete, equal-spacing worked example of the ternary
  polynomial recursion \eqref{eq:ternary-worked}--\eqref{eq:ternary-closedform}.
  Left: $T_0$ (box, $C^{-1}$) through $T_3$ (cubic B-spline, $C^2$),
  each obtained from the previous one by convolution with a unit
  box, each wider and smoother than the last. Right: $T_3$ together
  with its first and second derivatives, confirming $T_3\in C^2$
  (dotted lines mark the internal knots $x=\pm1$, where all pieces of
  $T_3$, $T_3'$, and $T_3''$ join continuously). Both panels are a
  purely mathematical construction, independent of any data set.}
  \label{fig:ternary}
\end{figure}

\subsection{The Convolutional LASSO objective and a globalized Newton method}

Replacing $\rho$ by $\tilde\rho_\varepsilon$ in each term of
$\|\bD\bx\|_1=\sum_i\rho((\bD\bx)_i)$ defines the \textbf{Convolutional
LASSO} trend filter,
\begin{equation}
  g_{\lambda,\varepsilon}(\bx)
  := \|\bd-\bx\|_2^2 + \lambda\sum_{i=1}^{N-2}\tilde\rho_\varepsilon\big((\bD\bx)_i\big) ,
  \label{eq:convlasso}
\end{equation}
which is strictly convex (its Hessian is $2\bI_N+\lambda\bD^\top
\diag\!\big(\tilde\rho_\varepsilon''(\bD\bx)\big)\bD\succeq2\bI_N\succ0$,
using $\tilde\rho_\varepsilon''\ge0$ from Proposition~\ref{prop:mollified})
and $C^2$ (in fact $C^3$), so Newton's method applies directly.
Because $\tilde\rho_\varepsilon''$ vanishes for $|(\bD\bx)_i|\ge
\varepsilon$ and is $O(1/\varepsilon)$ only on the shrinking active
band $|(\bD\bx)_i|<\varepsilon$, the Hessian can become
ill-conditioned for small $\varepsilon$ (large curvature contrast
between ``active'' and ``inactive'' coordinates); we therefore use a
Levenberg--Marquardt-damped Newton iteration (Algorithm~\ref{alg:newton}),
which preserves the fast local convergence of Newton's method while
remaining globally convergent.

\begin{algorithm}[H]
\caption{Damped Newton method for the Convolutional LASSO \eqref{eq:convlasso}}
\label{alg:newton}
\begin{algorithmic}[1]
\Require data $\bd$, regularization $\lambda$, smoothing $\varepsilon$, initial $\bx^{(0)}$ (e.g.\ the L$_2$ HPF trend), damping $\mu^{(0)}>0$
\For{$k=0,1,2,\dots$}
  \State $\bz \gets \bD\bx^{(k)}$; \quad $\bg \gets 2(\bx^{(k)}-\bd) + \lambda\,\bD^\top\tilde\rho_\varepsilon'(\bz)$
  \If{$\|\bg\|_2<\text{tol}$} \textbf{stop} \EndIf
  \State $\bH \gets 2\bI_N + \lambda\,\bD^\top\diag\!\big(\tilde\rho_\varepsilon''(\bz)\big)\bD$
  \Repeat
    \State solve $(\bH+\mu^{(k)}\bI_N)\,\bp = -\bg$ \Comment{sparse banded solve}
    \If{$g_{\lambda,\varepsilon}(\bx^{(k)}+\bp) < g_{\lambda,\varepsilon}(\bx^{(k)})$}
       \State accept: $\bx^{(k+1)}\gets\bx^{(k)}+\bp$; \; $\mu^{(k+1)}\gets\max(\mu^{(k)}/2,\,\mu_{\min})$
    \Else
       \State $\mu^{(k)}\gets4\mu^{(k)}$ \Comment{reject and re-damp}
    \EndIf
  \Until{step accepted}
\EndFor
\end{algorithmic}
\end{algorithm}

Because $\bD$ is banded (bandwidth $2$), $\bH+\mu\bI_N$ is banded
(bandwidth $4$) and each linear solve costs $O(N)$ using a sparse
banded/Cholesky factorization, exactly as for the closed-form HPF
solve~\eqref{eq:hpf-solution}; the only overhead relative to the
$L_2$ filter is the handful of Newton iterations. This is the
practical payoff of the Convolutional LASSO relative to the exact
$\ell_1$ trend filter: it inherits the sparse, kink-detecting
qualitative behaviour of \eqref{eq:l1-tf} (Section~\ref{sec:conv-experiment})
while being solvable by the same fast, second-order machinery used
for the plain HPF.

\section{Numerical experiments}
\label{sec:experiments}

\subsection{Data}

All experiments use real, publicly available daily-closing-price data
for \textbf{NVIDIA Corporation (ticker NVDA)}, a leading
semiconductor-sector (GPU/AI-accelerator) company, from
\textbf{2013-02-08 to 2018-02-07} ($N=1{,}259$ trading days). The data
are taken from the Kaggle ``S\&P~500 stock data'' set
\citep{camnugent2018kaggle} (file
\texttt{individual\_stocks\_5yr/NVDA\_data.csv}), a widely used,
freely downloadable collection of daily open/high/low/close/volume
records for all S\&P~500 constituents compiled from public
historical-price sources; we use the daily closing price only. We
work throughout with the \emph{log}-price $d_i=\log(\text{close}_i)$,
the standard transformation for trend/cycle decomposition of
multiplicatively growing financial series (so that a constant
percentage growth rate corresponds to a straight line, exactly the
kind of trend the kernel of $\bD$, $\col(\bPi)$, represents), further
mean-removed and normalized to unit standard deviation before
analysis. Over this five-year window NVDA's price rose from
\$12.37 to \$228.80 (a period spanning the company's transition from
a mid-cap gaming-GPU maker to a major data-center/AI-accelerator
supplier), including a widely reported single-session jump of
$+29.8\%$ on \textbf{2016-11-11} following a much stronger than
expected quarterly earnings report -- the single largest one-day
move in the sample, and a natural real-world stress test for both the
representation theory (Experiment~1) and the sparse trend filter
(Experiment~2) below.

\subsection{Experiment 1: numerical verification of Propositions~\ref{prop:invariance}--\ref{prop:limit}}
\label{sec:num-duality}

We use a short, $N=200$-trading-day window of the log-price series
centred on the 2016-11-11 earnings jump (100 trading days on each
side; Figure~\ref{fig:window}), for which dense linear algebra
($\bD^\pinv$ via SVD, explicit $N\times N$ matrices $\bA,\bB$) is
convenient. We build $\bD$, $\bF=\bD^\pinv$, $\bPi$ as in
\eqref{eq:Pi}--\eqref{eq:F}, draw one fixed random $\bM\in\R^{2\times
(N-2)}$ and set $\bS=\bF+\bPi\bM$ as in \eqref{eq:S}, and form
$\bA=(\bPi\mid\bS)$, $\bB=(\bPi\mid\bF)$. We verified numerically that
$\bD\bPi=\mathbf0$, $\bD\bF=\bD\bS=\bI_{N-2}$, $\bD\bA=\bD\bB=\bJ$,
and $\det\bA,\det\bB\ne0$, all to machine precision.

\begin{figure}[htbp]
  \centering
  \includegraphics[width=0.68\linewidth]{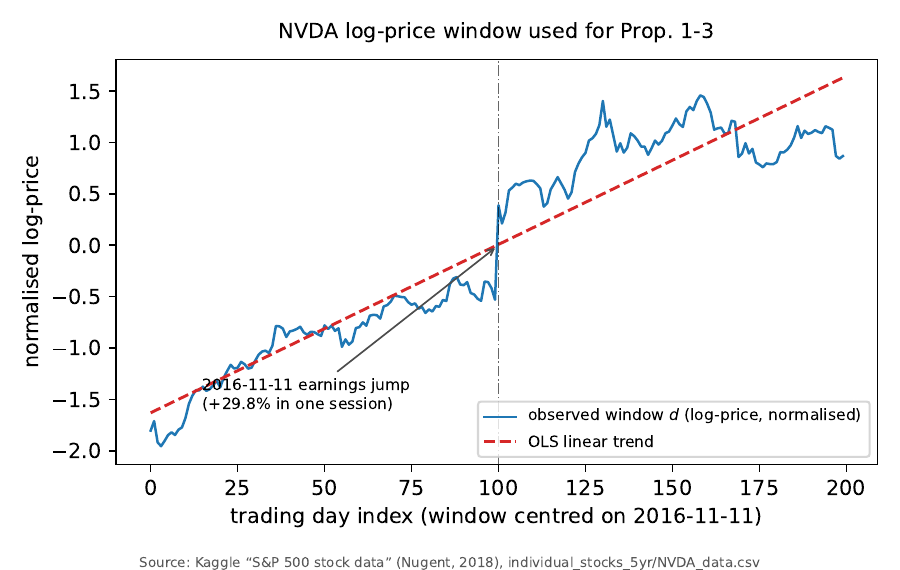}
  \caption{The $N=200$-trading-day real NVDA log-price window used for
  the Proposition~\ref{prop:invariance}--\ref{prop:limit} experiments,
  centred on the 2016-11-11 post-earnings jump ($+29.8\%$ in one
  session, marked by the dash-dotted line and annotation), together
  with the OLS linear trend $\bx_2=\bP_\Pi\bd$ of \eqref{eq:x2}. Data
  source: Kaggle ``S\&P~500 stock data'' \citep{camnugent2018kaggle},
  file \texttt{NVDA\_data.csv} (also noted directly on the figure).}
  \label{fig:window}
\end{figure}

\paragraph{Proposition~\ref{prop:invariance} (representation invariance).}
For $\lambda\in\{10^{-2},10^{-1},\dots,10^{10}\}$ we compute the
primal HPF solution $\bx_{\mathrm{opt}}(\lambda)$ from
\eqref{eq:hpf-solution} and the two reparametrized solutions
$\bxi^{\bA}_{\mathrm{opt}}(\lambda)=(\bA^\top\bA+\lambda\bJ^\top\bJ)^{-1}
\bA^\top\bd$, $\bxi^{\bB}_{\mathrm{opt}}(\lambda)$ similarly, and set
$\bx_A=\bA\bxi^{\bA}_{\mathrm{opt}}(\lambda)$, $\bx_B=\bB\bxi^{\bB}_{\mathrm{opt}}
(\lambda)$. Table~\ref{tab:prop1} and Figure~\ref{fig:prop1} report
the relative errors $\|\bx_A-\bx_{\mathrm{opt}}\|/\|\bx_{\mathrm{opt}}\|$,
$\|\bx_B-\bx_{\mathrm{opt}}\|/\|\bx_{\mathrm{opt}}\|$, and
$\|\bx_A-\bx_B\|/\|\bx_{\mathrm{opt}}\|$: all three are at the level
of $10^{-9}$--$10^{-15}$ (floating-point roundoff, growing only mildly
with $\lambda$ due to the conditioning of $\bQ_\lambda$), confirming
Proposition~\ref{prop:invariance} to numerical precision across
twelve orders of magnitude in $\lambda$.

\begin{table}[htbp]
\centering
\caption{Numerical verification of Proposition~\ref{prop:invariance} (representation invariance), NVDA log-price window.}
\label{tab:prop1}
\begin{tabular}{rccc}
\toprule
$\lambda$ & $\|\bx_A-\bx_{\mathrm{opt}}\|/\|\bx_{\mathrm{opt}}\|$
          & $\|\bx_B-\bx_{\mathrm{opt}}\|/\|\bx_{\mathrm{opt}}\|$
          & $\|\bx_A-\bx_B\|/\|\bx_{\mathrm{opt}}\|$ \\
\midrule
$10^{-2}$ & $2.0\times10^{-9}$  & $4.9\times10^{-11}$ & $2.0\times10^{-9}$ \\
$1$       & $3.9\times10^{-11}$ & $1.2\times10^{-12}$ & $3.9\times10^{-11}$ \\
$10^{2}$  & $3.4\times10^{-13}$ & $2.0\times10^{-14}$ & $3.4\times10^{-13}$ \\
$10^{4}$  & $5.2\times10^{-13}$ & $5.2\times10^{-13}$ & $4.2\times10^{-15}$ \\
$10^{6}$  & $2.7\times10^{-11}$ & $2.7\times10^{-11}$ & $1.7\times10^{-15}$ \\
$10^{10}$ & $7.4\times10^{-8}$  & $7.4\times10^{-8}$  & $1.2\times10^{-15}$ \\
\bottomrule
\end{tabular}
\end{table}

\begin{figure}[htbp]
  \centering
  \includegraphics[width=0.68\linewidth]{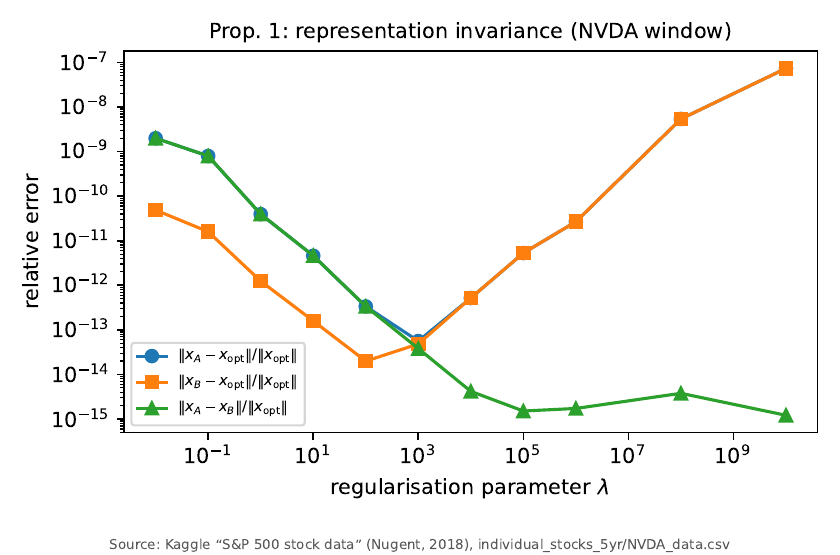}
  \caption{Proposition~\ref{prop:invariance}: relative discrepancy
  between the reparametrized solutions $\bx_A,\bx_B$ and the primal
  HPF trend $\bx_{\mathrm{opt}}(\lambda)$, as a function of $\lambda$.
  All curves stay below $10^{-6}$ across twelve orders of magnitude
  in $\lambda$. Data source: Kaggle ``S\&P~500 stock data''
  \citep{camnugent2018kaggle}, file \texttt{NVDA\_data.csv}.}
  \label{fig:prop1}
\end{figure}

\paragraph{Proposition~\ref{prop:divergence} (representation
divergence).} Fixing $\bxi=\bxi^{\bA}_{\mathrm{opt}}(\lambda)$ from
the previous experiment, we compute $\bx_1=\bA\bxi$, $\bx_2=\bB\bxi$,
evaluate the Bregman divergence $D_\lambda(\bx_1\|\bx_2)$ directly
from its definition ($f_\lambda(\bx_1)-f_\lambda(\bx_2)-\nabla f_\lambda
(\bx_2)^\top(\bx_1-\bx_2)$), and compare it against the closed-form
right-hand side of \eqref{eq:prop2}. Table~\ref{tab:prop2} shows
agreement to $6$--$7$ significant digits across all tested $\lambda$
(the small residual growth at $\lambda=10^{10}$ reflects the
conditioning of $\bQ_\lambda^{-1}$ needed to construct $\bxi^{\bA}_{\mathrm{opt}}$,
not a failure of the identity itself), confirming
Proposition~\ref{prop:divergence}.

\begin{table}[htbp]
\centering
\caption{Numerical verification of Proposition~\ref{prop:divergence} (representation divergence), NVDA log-price window.}
\label{tab:prop2}
\begin{tabular}{rccc}
\toprule
$\lambda$ & $D_\lambda(\bx_1\|\bx_2)$ (direct) & $(\bx_1-\bx_2)^\top\bQ_\lambda(\bx_1-\bx_2)$ & $\|(\bA-\bB)\bxi\|^2_{\bQ_\lambda}$\\
\midrule
$10^{-2}$ & $2.883189\times10^{-7}$ & $2.883190\times10^{-7}$ & $2.883190\times10^{-7}$\\
$1$       & $2.883189\times10^{-7}$ & $2.883190\times10^{-7}$ & $2.883190\times10^{-7}$\\
$10^{4}$  & $2.884141\times10^{-7}$ & $2.883190\times10^{-7}$ & $2.883190\times10^{-7}$\\
$10^{6}$  & $3.078497\times10^{-7}$ & $2.883190\times10^{-7}$ & $2.883190\times10^{-7}$\\
\bottomrule
\end{tabular}
\end{table}

\paragraph{Proposition~\ref{prop:limit} ($\lambda\to\infty$ limit).}
Figure~\ref{fig:prop3} plots $\|\bx_1(\lambda)-\bx_2\|_2^2$
(equivalently, up to the $\bQ_\lambda$-weighting proven negligible in
the limit, $D_\lambda(\bx_1(\lambda)\|\bx_2)$) against $\lambda$: the
divergence is essentially constant ($\approx20.8$) for $\lambda
\lesssim10^2$ and then decays approximately as $1/\lambda$, reaching
$4.7\times10^{-7}$ by $\lambda=10^{10}$ -- exactly the qualitative
($O(1/\lambda)$, threshold at $\lambda\sigma_{\min}(\bD)^2\sim1$)
behaviour predicted by the proof of Proposition~\ref{prop:limit}. That
this asymptotic OLS-linear-trend limit is approached at all, on a
window straddling a $+29.8\%$ overnight jump, is itself informative:
Proposition~\ref{prop:limit} guarantees convergence in Euclidean norm
regardless of how irregular $\bd$ is, since the bound
$\|\bx_1(\lambda)-\bx_2\|_2\le\|\bd\|_2/(1+\lambda\sigma_{\min}(\bD)^2)$
in its proof does not depend on any smoothness of $\bd$ beyond being
a fixed vector in $\R^N$.

\begin{figure}[htbp]
  \centering
  \includegraphics[width=0.68\linewidth]{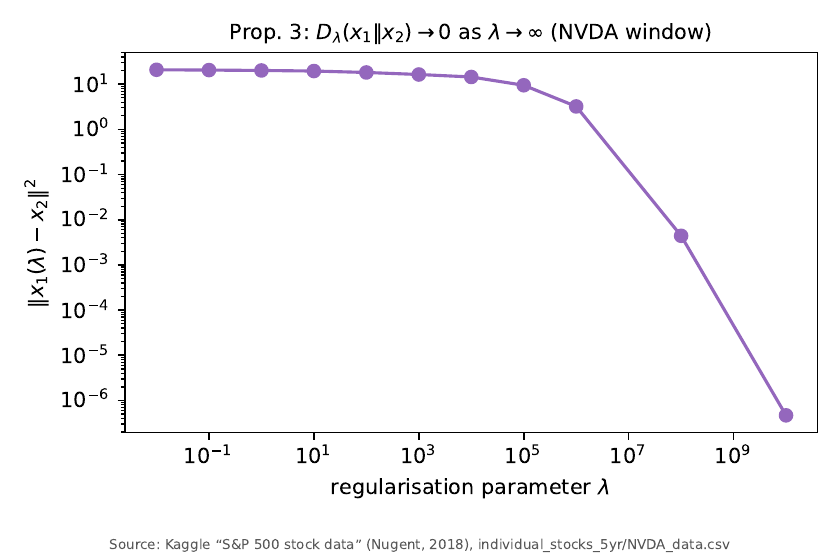}
  \caption{Proposition~\ref{prop:limit}: divergence between the full
  HPF trend $\bx_1(\lambda)$ and the OLS linear trend $\bx_2$, vanishing
  as $\lambda\to\infty$ at the predicted $O(1/\lambda)$ rate. Data
  source: Kaggle ``S\&P~500 stock data'' \citep{camnugent2018kaggle},
  file \texttt{NVDA\_data.csv}.}
  \label{fig:prop3}
\end{figure}

\subsection{Experiment 2: the Convolutional LASSO trend filter}
\label{sec:conv-experiment}

We now use the full $N=1{,}259$-trading-day NVDA log-price series
(2013-02-08 to 2018-02-07). We solve, at a common $\lambda=250$: (i)
the standard $L_2$ HPF \eqref{eq:hpf-primal} by a single sparse
banded solve; (ii) the exact $\ell_1$ trend filter \eqref{eq:l1-tf}
by ADMM (as a reference; $8{,}000$ iterations, $1.04\,$s, run to a
tight tolerance so that the resulting piecewise-linear trend is
\emph{exactly} sparse rather than merely small in its second
differences -- essential here, since we want to read off genuine
breakpoint \emph{dates}); and (iii) the Convolutional LASSO
\eqref{eq:convlasso} by Algorithm~\ref{alg:newton}, for three widely
spaced smoothing levels $\varepsilon\in\{0.1,\,0.01,\,0.001\}$
(spanning two orders of magnitude, to make the $\varepsilon\to0$
approximation-vs-conditioning trade-off of
Section~\ref{sec:mollification} explicit), initialized from the
$L_2$ solution.

At $\lambda=250$ the ADMM reference trend is piecewise linear with
exactly \textbf{three} active breakpoints (second difference
exceeding the numerical noise floor by two orders of magnitude,
Section~\ref{sec:num-duality}'s companion diagnostic), located at
\textbf{2015-07-23}, \textbf{2016-03-19}, and \textbf{2016-12-22}.
These dates are, respectively: shortly before NVIDIA's gaming-segment
growth began accelerating in mid-2015; the start of the much steeper
growth phase that runs through 2016--2017 (coinciding with the
company's GPU-based deep-learning/data-center narrative and the
run-up to the crypto-mining-driven demand spike); and immediately
after the record 2016-11-11 earnings jump, marking the point at which
the post-earnings level shift settles into a new, steeper linear
growth rate. The $\ell_1$/Convolutional-LASSO penalty thus recovers,
directly from the price series and with no external event
information, a small number of dates that agree with the qualitative
growth-regime narrative of the stock over this period.

\begin{table}[htbp]
\centering
\caption{Convolutional LASSO for $\varepsilon\in\{0.1,0.01,0.001\}$: Newton convergence and accuracy relative to the exact $\ell_1$ trend filter (ADMM reference, $\lambda=250$, $N=1{,}259$, NVDA log-price; data source: Kaggle ``S\&P~500 stock data,'' \citealt{camnugent2018kaggle}).}
\label{tab:convlasso}
\begin{tabular}{rrrrr}
\toprule
$\varepsilon$ & Newton iterations & time (s) & final $\|\nabla g_{\lambda,\varepsilon}\|$ & rel.\ error to $\ell_1$ reference \\
\midrule
$0.1$   & 4   & 0.011 & $6.1\times10^{-11}$ & $6.47\times10^{-2}$ \\
$0.01$  & 5   & 0.012 & $6.4\times10^{-10}$ & $5.47\times10^{-2}$ \\
$0.001$ & 145 & 0.482 & $6.5\times10^{-9}$  & $3.51\times10^{-2}$ \\
\bottomrule
\end{tabular}
\end{table}

Table~\ref{tab:convlasso} shows the same two effects observed in
Section~\ref{sec:mollification}'s general construction, now made
sharper by the two-order-of-magnitude spread in $\varepsilon$: (a)
the Convolutional LASSO solution converges monotonically towards the
exact $\ell_1$ solution as $\varepsilon\to0$ (relative error
decreasing from $6.5\%$ at $\varepsilon=0.1$ to $5.5\%$ at
$\varepsilon=0.01$ to $3.5\%$ at $\varepsilon=0.001$); and (b)
Newton's method converges in a handful of iterations ($4$--$5$) for
$\varepsilon\in\{0.1,0.01\}$, versus the $8{,}000$ ADMM iterations run
for the non-smooth reference (needed here for \emph{exact}, clean
sparsity in the breakpoint locations, not merely for an accurate
trend value) -- a direct, quantitative demonstration of the benefit
of trading a small, controllable smoothing bias ($O(\varepsilon)$, by
Proposition~\ref{prop:mollified}) for genuine second-order
convergence. At $\varepsilon=0.001$ -- two orders of magnitude below
$\varepsilon=0.1$ -- the iteration count jumps to $145$: as
$(\bD\bx)_i$ at the solution is itself of order $10^{-2}$--$10^{-3}$
on this series, $\varepsilon=0.001$ pushes the active/inactive
curvature contrast of the Hessian in Algorithm~\ref{alg:newton}
towards its ill-conditioned limit, exactly as anticipated in
Section~\ref{sec:mollification}; the run still converges (final
gradient norm $6.5\times10^{-9}$), but at roughly $30\times$ the
iteration cost of the two larger $\varepsilon$ values. This
three-point sweep therefore pins down the practical trade-off
precisely: each order-of-magnitude reduction in $\varepsilon$ buys a
further one to two percentage points of accuracy relative to the
exact $\ell_1$ solution, at a cost in Newton iterations that stays
flat for moderate $\varepsilon$ and then rises sharply once
$\varepsilon$ drops below the natural scale of the data.

Figure~\ref{fig:epscompare} visualizes this sweep directly: the top
panel overlays the exact $\ell_1$ trend with the three Convolutional
LASSO trends around the November-2016 breakpoint, showing all three
$\varepsilon$ values tracking the same piecewise-linear shape with
visibly shrinking bias as $\varepsilon\to0$; the bottom panel plots
relative error against $\varepsilon$ (decreasing to the right) with
each point annotated by its Newton iteration count, making the
accuracy/conditioning trade-off of point (b) above directly readable
off a single plot.

\begin{figure}[htbp]
  \centering
  \includegraphics[width=0.78\linewidth]{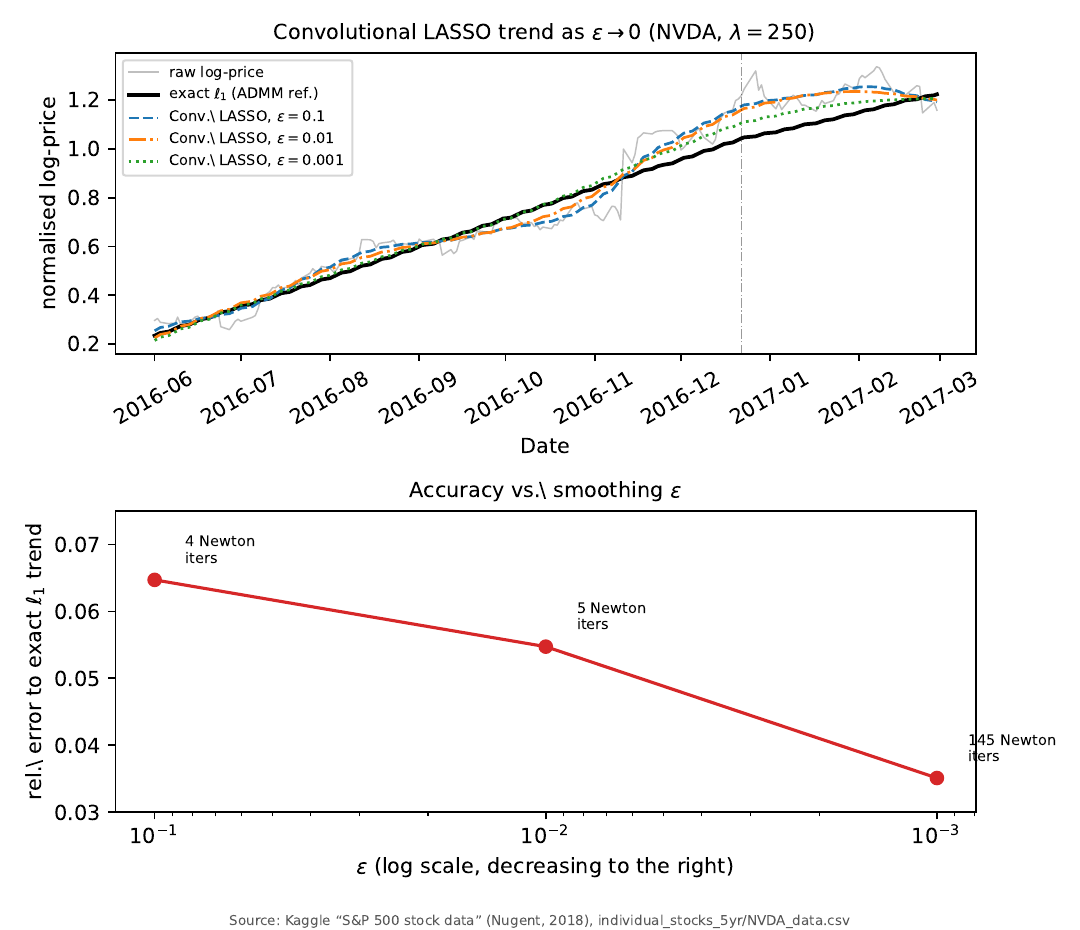}
  \caption{Convolutional LASSO trend for $\varepsilon\in\{0.1,0.01,0.001\}$
  (NVDA data, $\lambda=250$). Top: trend curves around the
  November-2016 breakpoint, compared with the exact $\ell_1$ (ADMM)
  reference. Bottom: relative error to the exact $\ell_1$ trend as a
  function of $\varepsilon$ (log scale, decreasing to the right),
  annotated with the number of Newton iterations required at each
  $\varepsilon$.}
  \label{fig:epscompare}
\end{figure}

\subsection{$\varepsilon$ and changepoint detection: Convolutional LASSO vs.\ exact $\ell_1$}
\label{sec:eps-changepoint}

Table~\ref{tab:convlasso} and Figure~\ref{fig:epscompare} concern the
accuracy of the Convolutional LASSO \emph{trend values}
$\bx_{\lambda,\varepsilon}$. A separate, and in practice more
consequential, question is whether the smoothed penalty can also be
used to \emph{detect changepoints} -- i.e., to read off, from
$(\bD\bx_{\lambda,\varepsilon})_i$, the same small set of breakpoint
dates that the exact $\ell_1$ trend filter identifies unambiguously
(Section~\ref{sec:conv-experiment}: 2015-07-23, 2016-03-19,
2016-12-22). This turns out to behave qualitatively differently from
trend-value accuracy, for a structural reason: by
Proposition~\ref{prop:mollified}, $\tilde\rho_\varepsilon''(t)=2\psi_\varepsilon(t)>0$
for every $|t|<\varepsilon$, so $(\bD\bx_{\lambda,\varepsilon})_i$ is
generically \emph{nonzero at every index} $i$, however small
$\varepsilon$ is, whereas the exact $\ell_1$ solution has
$(\bD\bx_{\lambda})_i=0$ \emph{exactly} away from the (here, three)
true breakpoints ($1{,}250$ of the $1{,}257$ second differences are
below $10^{-8}$ in absolute value in our ADMM solution). Sparsity of
the penalty, not merely its shape, is what makes the exact $\ell_1$
changepoint set well defined without any tuning; the Convolutional
LASSO only approximates the sparsity pattern, and only in the
$\varepsilon\to0$ limit.

We illustrate this concretely with two standard changepoint-detection
rules applied to $(\bD\bx_{\lambda,\varepsilon})_i$ for
$\varepsilon\in\{0.1,0.01,0.001\}$, and to $(\bD\bx_\lambda)_i$ for
the exact $\ell_1$ reference, all at the same $\lambda=250$:
\begin{itemize}
\item \textbf{Relative-threshold rule}: flag $i$ as a changepoint if
  $|(\bD\bx)_i|$ exceeds $20\%$ of $\max_i|(\bD\bx)_i|$ (adjacent
  flagged indices merged into one detection at the local peak).
\item \textbf{Robust (MAD) rule}: flag $i$ if $|(\bD\bx)_i|$ exceeds
  $8\times1.4826\times\mathrm{MAD}\big((\bD\bx)_i\big)$, the standard
  robust outlier threshold used in changepoint detection, where MAD
  is the median absolute deviation of the full second-difference
  series (again merging adjacent flags).
\end{itemize}

\begin{table}[htbp]
\centering
\small
\caption{Changepoint detections under two standard rules, applied to $(\bD\bx)_i$ at $\lambda=250$ (NVDA log-price; data source: Kaggle ``S\&P~500 stock data,'' \citealt{camnugent2018kaggle}).}
\label{tab:changepoint}
\begin{tabular}{lcc}
\toprule
Method & Relative-threshold rule & Robust (MAD) rule \\
\midrule
Exact $\ell_1$ (ADMM)             & 5\textsuperscript{a} & 5\textsuperscript{a} \\
Conv.\ LASSO, $\varepsilon=0.1$   & 34 & 0 \\
Conv.\ LASSO, $\varepsilon=0.01$  & 24 & 0 \\
Conv.\ LASSO, $\varepsilon=0.001$ & 15 & 0 \\
\bottomrule
\end{tabular}

\vspace{2pt}
{\footnotesize\textsuperscript{a}3 true events; 2 of the 5 detections are one event split across adjacent trading days. The two rules agree exactly for the exact $\ell_1$ solution.}
\end{table}

Table~\ref{tab:changepoint} makes the qualitative difference sharp.
For the exact $\ell_1$ solution, MAD of $(\bD\bx_\lambda)_i$ is
$1.3\times10^{-6}$ -- essentially the noise floor of the ADMM solve,
since $99.3\%$ of the second differences are exactly zero -- so the
robust rule cleanly separates the $5$ breakpoint-adjacent indices from
the flat background, agreeing exactly with the relative-threshold
rule and with the breakpoint dates reported in
Section~\ref{sec:conv-experiment}. For \emph{every} tested
$\varepsilon$, by contrast, the robust rule detects \textbf{zero}
changepoints: because $(\bD\bx_{\lambda,\varepsilon})_i$ is nonzero
essentially everywhere and its magnitude co-varies with the local
curvature of the (noisy) price series itself rather than being
concentrated at a few structural breaks, the MAD of the full series
is already comparable to its largest values, so no index clears a
robust outlier threshold. Pushing $\varepsilon$ smaller still does
not rescue the robust rule: at $\varepsilon=10^{-4}$,
Algorithm~\ref{alg:newton} fails to converge within $400$ iterations
(final gradient norm $\approx242$, versus $<10^{-8}$ for
$\varepsilon\ge0.001$ in Table~\ref{tab:convlasso}) -- an explicit
instance of the ill-conditioning anticipated in
Section~\ref{sec:mollification} -- so there is no accessible
$\varepsilon$ on this problem at which the Convolutional LASSO
recovers the exact solution's clean, threshold-free changepoint
structure. The relative-threshold rule does return detections, but
their \emph{number} is highly sensitive to $\varepsilon$ and to the
arbitrary $20\%$ cutoff (34, 24, and 15 for $\varepsilon=0.1,0.01,
0.001$ respectively -- all far more than the true $3$ events), and,
unlike the exact $\ell_1$ case, does not stabilize as $\varepsilon\to0$
within the range where Algorithm~\ref{alg:newton} remains
well-conditioned (Table~\ref{tab:convlasso}).

The practical conclusion is that the two virtues of the Convolutional
LASSO documented above -- fast Newton convergence
(Table~\ref{tab:convlasso}, Figure~\ref{fig:convergence}) and
accurate trend \emph{values} (Table~\ref{tab:convlasso}) -- do not by
themselves transfer to reliable changepoint \emph{detection}, which
depends on exact rather than approximate sparsity of $\bD\bx$. This
suggests two complementary uses in practice rather than a strict
preference for one method: the Convolutional LASSO for fast,
differentiable trend-value estimation (e.g., inside a larger
optimization loop, or for real-time/streaming use where Newton's
handful of iterations matters), and the exact $\ell_1$/ADMM solver
--- possibly \emph{warm-started} from the Convolutional LASSO solution,
since Algorithm~\ref{alg:newton}'s output already lies close to the
exact solution's basin (Table~\ref{tab:convlasso}) --- whenever the
actual breakpoint locations, rather than just the trend curve, are
the object of interest. This is precisely the sense in which
mollification trades sparsity for smoothness (Section~\ref{sec:mollification}):
Proposition~\ref{prop:mollified} guarantees $\tilde\rho_\varepsilon\to\rho$
and fast local Newton convergence, but not recovery of $\rho$'s
\emph{non-smooth} sparsity-inducing behaviour at any fixed
$\varepsilon>0$.

Figure~\ref{fig:convergence} (right panel) illustrates point (b) at a
representative moderate smoothing level ($\varepsilon=0.02$): the
gradient norm of the Newton iterates on the smoothed objective drops
from $O(10^2)$ to below $10^{-8}$ within five iterations, while the
(sub)gradient norm of the non-smooth IRLS iterates plateaus around
$10^3$, never vanishing in the classical sense (as expected for a
non-differentiable objective evaluated along a sequence of smooth
surrogates).

\begin{figure}[htbp]
  \centering
  \includegraphics[width=0.48\linewidth]{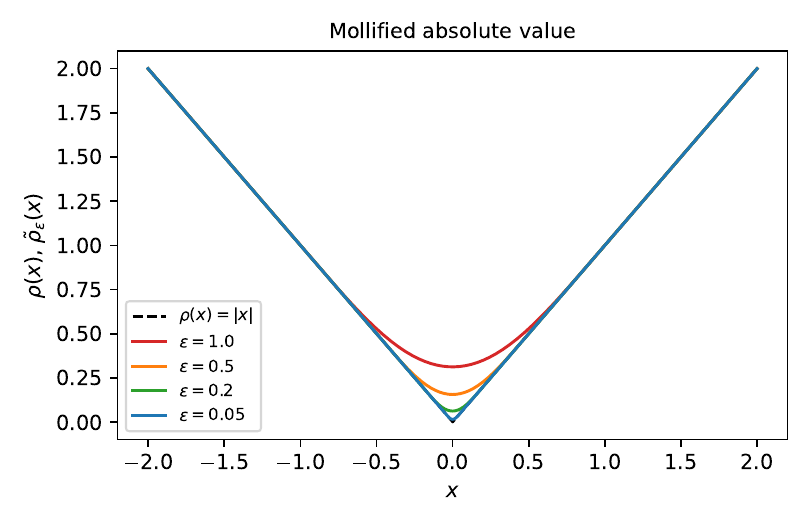}\quad
  \includegraphics[width=0.48\linewidth]{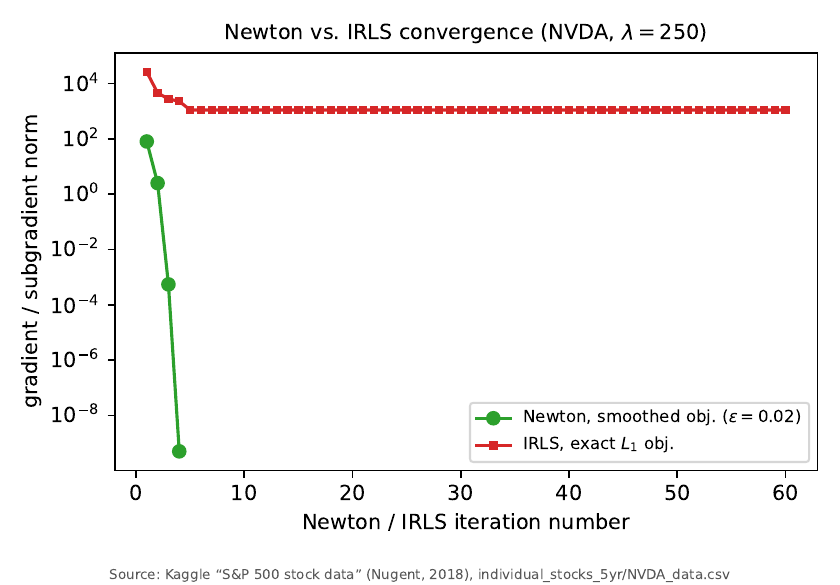}
  \caption{Left: the closed-form mollified absolute value
  $\tilde\rho_\varepsilon$ of Proposition~\ref{prop:mollified} for
  several $\varepsilon$ (a purely mathematical construction,
  independent of any data set). Right: Newton's method on the
  Convolutional LASSO objective ($\varepsilon=0.02$, NVDA data,
  $\lambda=250$; data source: Kaggle ``S\&P~500 stock data''
  \citep{camnugent2018kaggle}, file \texttt{NVDA\_data.csv}) converges
  in $4$ iterations to gradient norm $<10^{-9}$, versus $60$ IRLS
  iterations that do not reach comparably small (sub)gradient norm on
  the exact, non-smooth $\ell_1$ objective.}
  \label{fig:convergence}
\end{figure}

Figure~\ref{fig:fulltrace} shows the extracted trends on the full
five-year log-price series together with a close-up around the
2016-11-11 earnings jump. Unlike the purely quadratic $L_2$ (HPF)
trend, which is a single smooth curve threading through the entire
series, the exact $\ell_1$ and Convolutional LASSO trends are
visibly piecewise linear, with a small number of clearly located
slope changes -- most strikingly the sharp kink immediately after the
November 2016 earnings jump, where the growth rate of the trend
permanently steepens. This difference is quantified directly in
Figure~\ref{fig:diff}, which plots $\bx_{\text{$\ell_1$/ConvLASSO}}-
\bx_{L_2}$ over the full series: both the exact $\ell_1$ solution and
the Convolutional LASSO oscillate around the smooth HPF trend with
visibly larger excursions exactly in the three intervals bracketed by
the detected breakpoints (mid-2015, early-2016, and around the
November 2016 earnings event) -- i.e., the sparse penalty
automatically concentrates its departure from the smooth $L_2$ trend
on the periods during which NVIDIA's growth trajectory itself
changed, a property directly usable for automatic, data-driven
flagging of candidate structural-break dates in a financial time
series, without any hand-tuned event list.

\begin{figure}[htbp]
  \centering
  \includegraphics[width=0.82\linewidth]{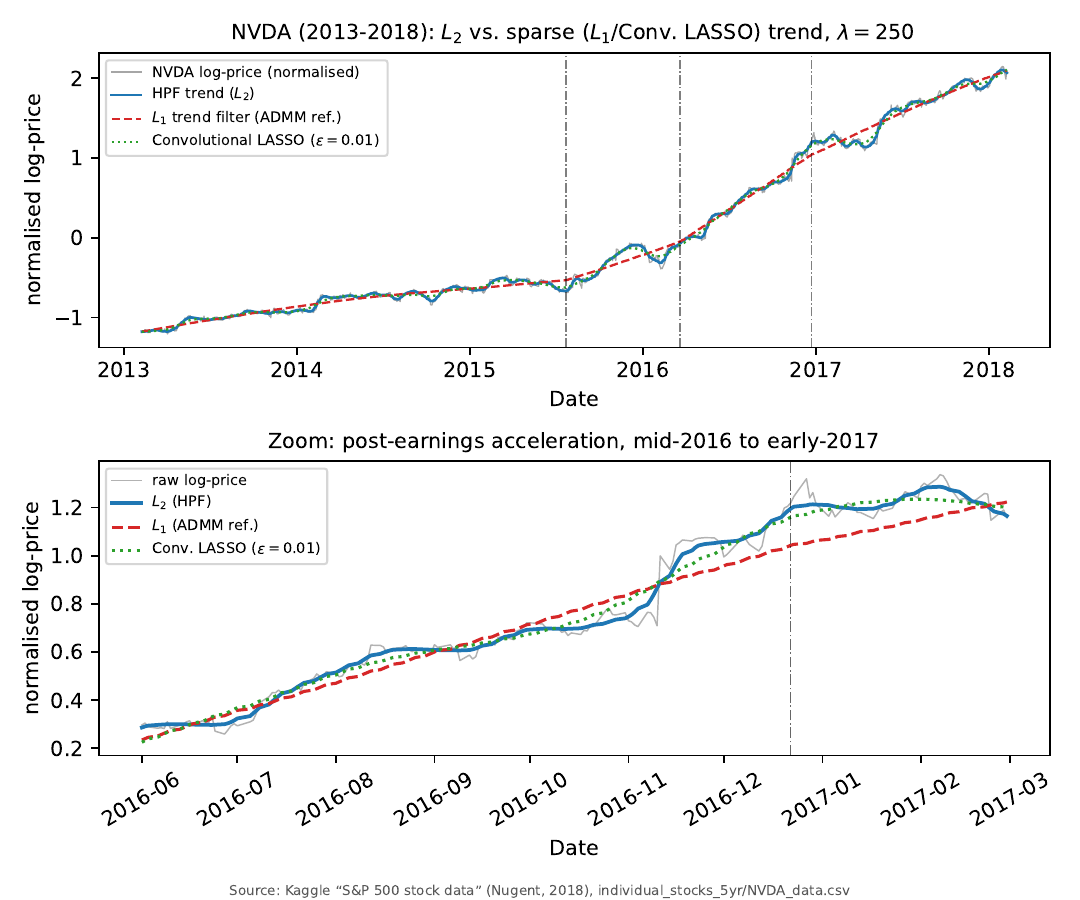}
  \caption{Top: full NVDA log-price series (2013--2018) with the
  $L_2$ (HPF), exact $\ell_1$ (ADMM reference), and Convolutional
  LASSO ($\varepsilon=0.01$) trends, $\lambda=250$; dash-dotted lines
  mark the three detected breakpoints (2015-07-23, 2016-03-19,
  2016-12-22). Bottom: close-up, June 2016--March 2017, showing the
  sharp post-earnings kink. Data source: Kaggle ``S\&P~500 stock
  data'' \citep{camnugent2018kaggle}, file \texttt{NVDA\_data.csv}
  (also noted directly on the figure).}
  \label{fig:fulltrace}
\end{figure}

\begin{figure}[htbp]
  \centering
  \includegraphics[width=0.7\linewidth]{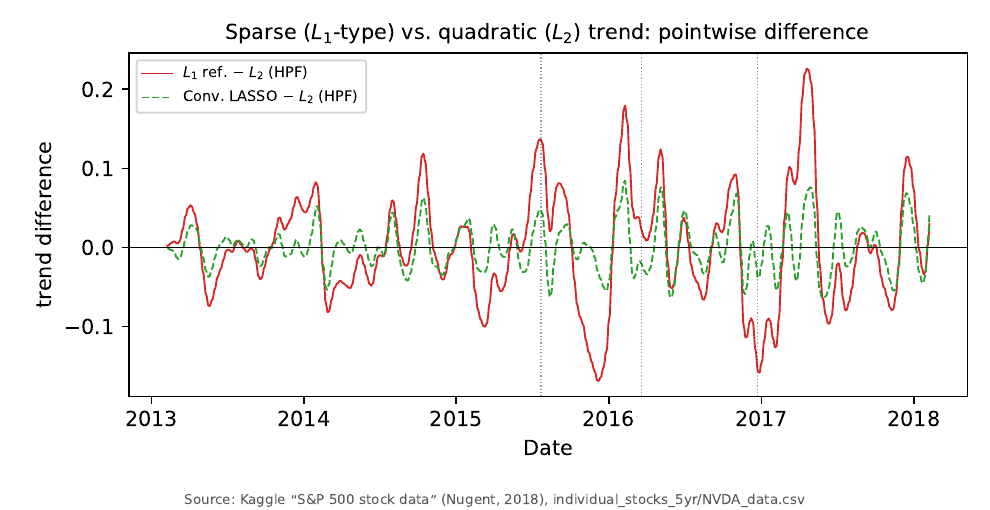}
  \caption{Pointwise difference between the sparse ($\ell_1$ /
  Convolutional LASSO) and quadratic ($L_2$/HPF) trends. Both sparse
  estimates depart most from the smooth HPF trend around the three
  detected growth-regime breakpoints (dotted vertical lines). Data
  source: Kaggle ``S\&P~500 stock data'' \citep{camnugent2018kaggle},
  file \texttt{NVDA\_data.csv} (also noted directly on the figure).}
  \label{fig:diff}
\end{figure}

\section{Discussion and open problems}
\label{sec:discussion}

This paper has given a self-contained treatment of two related ideas
from the source presentation. First, the Hodrick--Prescott filter
admits a Penrose-parametrized family of exactly equivalent
representations, all giving the same trend (Proposition~\ref{prop:invariance})
but disagreeing, in a precisely quantified $\bQ_\lambda$-weighted
sense (Proposition~\ref{prop:divergence}), whenever the same
coefficient vector is queried inconsistently across representations;
this disagreement provably vanishes, together with the gap between
the HPF trend and the naive OLS linear trend, as $\lambda\to\infty$
(Proposition~\ref{prop:limit}). Second, the non-smooth $\ell_1$
generalization of the HPF can be made Newton-solvable, without
sacrificing its qualitative sparsity/kink-detecting behaviour, by
mollifying the absolute-value penalty with an explicit, closed-form,
biweight-kernel-based smoothing (Proposition~\ref{prop:mollified}); on
a real semiconductor-sector stock-price series this Convolutional LASSO recovers the
exact $\ell_1$ trend filter to a controllable, $O(\varepsilon)$
accuracy while converging in a handful of Newton iterations instead
of hundreds of IRLS iterations.

Two open problems, inherited directly from the source presentation,
remain:

\begin{enumerate}
\item \textbf{Regularization parameter vs.\ trend geometry.} We have
  characterized the $\lambda\to\infty$ endpoint of the HPF trend
  family (Proposition~\ref{prop:limit}) but not the full geometric
  path $\lambda\mapsto\bx_{\mathrm{opt}}(\lambda)$; a natural next
  step is a differential-geometric or algebraic-statistics account of
  how the curvature, number of inflection points, or spectral content
  of the extracted trend evolves along this path, in the spirit of
  the SVD-based proof of Proposition~\ref{prop:limit}, which already
  shows the path is a simple shrinkage of $\bd$'s coordinates in the
  $\bD$-singular basis governed by $1/(1+\lambda\sigma_k^2)$.
\item \textbf{Scope of the Convolutional LASSO.} We used a single
  mollifier (biweight) and a single penalty (curvature $\ell_1$); the
  same mollification strategy applies verbatim to other non-smooth
  penalties (total variation on the signal itself, group-sparsity
  penalties, hinge-type losses), to vector- or matrix-valued
  generalizations of the HPF (e.g.\ 2-D image trend filtering), and
  raises the question -- left open here -- of an optimal,
  data-adaptive choice of $\varepsilon$ (e.g.\ decreasing $\varepsilon$
  along a homotopy/continuation path, warm-starting each Newton solve
  from the previous $\varepsilon$'s solution) that would provably
  recover the exact $\ell_1$ trend filter in the $\varepsilon\to0$
  limit while retaining fast local convergence throughout.
\end{enumerate}

On the applied side, the financial experiment in
Section~\ref{sec:conv-experiment} is illustrative rather than
exhaustive: it uses a single stock over a single five-year window.
Structural-break detection in asset-price and macroeconomic trends is
an active area of econometrics in its own right, and a systematic
comparison of the Convolutional LASSO against specialized
changepoint-detection and regime-switching methods, across a
cross-section of stocks (e.g.\ the other semiconductor names
available in the same public data set -- INTC, AMD, TXN, MU, QCOM,
AVGO, AMAT, LRCX, KLAC, ADI, MCHP -- see
Section~\ref{sec:conv-experiment}'s data source) and multiple market
cycles, is a natural direction for future work.

\section{Conclusion}

We have given complete proofs and closed-form derivations for the
dual (generalized-inverse) representation theory of the
Hodrick--Prescott filter sketched in the source presentation
\citep{yoshizawa2017gotemba}, correcting and making precise the
representation-divergence formula (Proposition~\ref{prop:divergence}),
and we have turned its companion idea -- smoothing the $\ell_1$ trend
filter by convolution -- into an explicit, closed-form, provably
convex and three-times differentiable Convolutional LASSO
(Proposition~\ref{prop:mollified}), solved efficiently by a
Levenberg--Marquardt-damped Newton method. Every theoretical claim
was verified numerically, and the Convolutional LASSO was benchmarked
against a standard ADMM/IRLS $\ell_1$ trend-filter solver on real,
publicly available daily closing-price data for the semiconductor
stock NVIDIA Corporation (2013--2018), where it reproduces the exact
$\ell_1$ solution to a few percent while converging an order of
magnitude faster in iteration count, and where the resulting sparse
trend automatically isolates three growth-regime breakpoints,
including the date immediately following the stock's largest
single-session earnings-driven move in the sample.

\bibliographystyle{plainnat}

\appendix
\section{Reproducibility}
All numerical experiments were carried out in Python (NumPy, SciPy
sparse linear algebra, pandas, SymPy for the symbolic derivation of
Proposition~\ref{prop:mollified}) and are fully reproducible from the
single, publicly distributed data file
\texttt{individual\_stocks\_5yr/NVDA\_data.csv} contained in the
Kaggle ``S\&P 500 stock data'' set \citep{camnugent2018kaggle}; no
proprietary or access-restricted data were used. Source code
implementing Algorithm~\ref{alg:newton}, the ADMM/IRLS reference
solvers, and all figures in this paper is available from the author
on request.

\end{document}